\documentclass[conference,letterpaper]{IEEEtran}

\usepackage[cmex10]{amsmath} % Use the [cmex10] option to ensure complicance
\usepackage{stfloats}
\usepackage[T1]{fontenc}
\usepackage{url}
\usepackage{ifthen}
\usepackage{cite}
\usepackage{url} 
\usepackage{hhline}
\usepackage{array}

\usepackage{nicematrix}
\usepackage{arydshln}
\usepackage{tikz}
\usepackage{amsmath}
\usepackage{amssymb}

\usepackage{graphicx}
\usepackage{graphics}
\usepackage{subcaption}
\usepackage{tabulary}
\usepackage{amsthm}
\usepackage{enumitem}
\usepackage{yfonts}
\usepackage{booktabs}
\usepackage{mfirstuc} 
\usepackage{caption}
\usepackage{subcaption}
\usepackage{lipsum}
\usepackage{setspace}
\newtheorem{theorem}{Theorem}
\newtheorem{lemma}{Lemma}
\newtheorem{definition}{Definition}
\newtheorem{example}{Example}

\newcommand{\mr}{\mathrm}
\newcommand{\mb}{\mathbf}
\newcommand{\mc}{\mathcal}

\begin{document}
\title{Symmetric Entropy Regions of Degrees Six and Seven} 

% %%% Single author, or several authors with same affiliation:
 \author{%
  \IEEEauthorblockN{Zihan Li, Shaocheng Liu, and Qi Chen}
 \IEEEauthorblockA{School of Telecommunications Engineering, Xidian University, China.  
                    \\
                  Email: \{lzih,lsc\}@stu.xidian.edu.cn, qichen@xidian.edu.cn
 }
}

%%% Several authors with up to three affiliations:
%\author{%
%  \IEEEauthorblockN{Zihan Li}
%  \IEEEauthorblockA{Department of Statistics and Data Science\\
%                    Yale University\\
%                    New Haven, CT, USA\\
%                    Email: andrew.barron@yale.edu}
%  \and
%  \IEEEauthorblockN{Claude E.~Shannon and David Slepian}
%  \IEEEauthorblockA{Bell Telephone Laboratories, Inc.\\ 
%                    Murray Hill, NJ, USA\\
%                    Email: \{csh, dsl\}@bell-labs.com}
%}

%%% Many authors with many affiliations:
% \author{%
%   \IEEEauthorblockN{Andrew R.~Barron\IEEEauthorrefmark{1},
%                     Claude E.~Shannon\IEEEauthorrefmark{2},
%                     David Slepian\IEEEauthorrefmark{2},
%                     and Jacob Ziv\IEEEauthorrefmark{2}\IEEEauthorrefmark{3}}
%   \IEEEauthorblockA{\IEEEauthorrefmark{1}%
%                    Department of Statistics and Data Science, Yale University, New Haven, CT, USA,
%                     andrew.barron@yale.edu}
%   \IEEEauthorblockA{\IEEEauthorrefmark{2}%
%                     Bell Telephone Laboratories, Inc.,
%                     Murray Hill, NJ, USA,
%                     \{csh,dsl,jz\}@bell-labs.com}
%   \IEEEauthorblockA{\IEEEauthorrefmark{3}%
%                     Department of Electrical Engineering, Technion---Institute of Technology, Haifa, Israel,
%                     jz@ee.technion.ac.il}
% }

\maketitle

%%%%%%
%% Abstract: 
%% If your paper is eligible for the student paper award, please add
%% the comment "THIS PAPER IS ELIGIBLE FOR THE STUDENT PAPER
%% AWARD." as a first line in the abstract. 
%% For the final version of the accepted paper, please do not forget
%% to remove this comment!
%%

\begin{abstract}
   In this paper, we classify all $G$-symmetric almost entropic regions according to their Shannon-tightness, that is, whether they can be fully characterized by Shannon-type inequalities, where $G$ is a permutation group of degree $6$ or $7$.
   
\end{abstract}

\section{Introduction}

Let $N=\{1,2,\cdots,n\}$ and $(X_i,i\in N)$ be a random vector indexed by $N$.
The set function $\mathbf{h}: 2^{N}\to \mathbb{R}$ defined by 
\begin{equation*}
	\mathbf{h}(A)=H(X_A), \quad A\subseteq N
\end{equation*}
is called the \emph{entropy function} of $(X_i,i\in N)$, 
The Euclidean space $\mc{H}_n \triangleq\ \mathbb{R}^{2^{N}}$ where entropy functions live in is called the \emph{entropy space} of degree $n$.
The set of all entropy functions, denoted by $\Gamma^*_n$, is called
the \emph{entropy region} of degree $n$,
and its closure $\overline{\Gamma_n^*}$, is called
the \emph{almost entropic region} of degree $n$ \cite{yeung2008information}.
%$\mc{H}_n=\mathbb{R}^{2^N}$. 
%For all set function $\mathbf{h} \in \mc{H}_n$,
%we call $\mb{h}$ entropic if there exists a set of $n$ jointly distributed discrete random variables 
%$X_N \triangleq  \{X_i:i \in N\}$ such that $\mathbf{h}(A)=H(X_A)$ for any $A \subseteq N$,
%where $H(X_\emptyset)=0$ by convention. So $\mb{h}$ is called the entropy function of $X_N$ and 
%$\mc{H}_n$ is called the entropy space of degree $n$. The set of all entropy functions in the region $\mc{H}_n$,
%is denoted by $\Gamma_n^*$, 

It is known that for all $A$, $B \subseteq N$, an entropy function satisfies the following three polymatroidal axioms:
\begin{align}
\mathbf{h}(A)&\geq 0 \label{a1}\\
\mathbf{h}(A)&\leq \mathbf{h}(B) \ \text{if}\  A\subseteq B\label{a2}\\ 
\mathbf{h}(A)+\mathbf{h}(B)&\geq \mathbf{h}(A\cap B)+ \mathbf{h}(A \cup B)\label{a3}. 
\end{align}
That is, every entropy function $\mathbf{h}$ is (the rank function of) 
a polymatroid \cite{fujishige1978polymatroidal}. 
It has been proved that the basic information inequalities, 
i.e., the nonnegativity of Shannon information measures, are equivalent to the polymatroidal axioms. 
So the region $\Gamma_n$ bounded by Shannon-type information inequalities is also called  \emph{polymatroidal region} of degree $n$ and $\Gamma_n^* \subseteq \Gamma_n$.
It has been proved that $\Gamma_2^* = \Gamma_2$, $ \Gamma_3^* \subsetneq \Gamma_3$
but $ \overline{\Gamma_3^*} = \Gamma_3$\cite{zhang1997non,matus2005piecewise,chen2012extremeray}.
Due to the existence of infinitely many linear non-Shannon-type inequalities \cite{Zhang1998},\cite{Matus2007infinite}, when $n \geq 4$, 
$\overline{\Gamma_n^*}$ is a non-polyhedral convex cone with
$\overline{\Gamma_n^*} \subsetneq \Gamma_n$, and the full characterizations of 
$\Gamma_n^*$ and $\overline{\Gamma_n^*}$ are extremely difficult.

By imposing symmetric constraints on entropy regions, Chen and Yeung \cite{partition} introduced symmetric entropy functions, and they proved that for the closure of a partition-symmetric entropy region, it can be fully characterized by Shannon-type inequalities if and only if the partition is the 1-partition or a 2-partition with one of its blocks being a singleton. In \cite{Symmetries16ITW}, Apte, Chen and Walsh classifies all $G$-symmetric almost entropic regions according to whether they are equal to the  $G$-symmetric polymatroidal regions, where $G$ is a permutation group of degree $4$ or $5$.

In this paper, we do the same classification to the $G$-symmetric almost entropic regions when $G$ are permutation groups of degree $6$ or $7$.

The rest of this paper is organized as follows. Section \ref{pre} gives the preliminaries on 
orbit structures of permutation groups acting on $2^N$ and symmetric entropy functions. The main theorems of the paper about the classifications of the G-symmetric almost entropic regions according to whether they can be fully characterized by Shannon information inequalities are in Section \ref{results}.
\section{Preliminaries}
\label{pre}
In this section, we give brief preliminaries on  orbit structures of permutation groups acting on $2^N$, and the symmetries in the entropy space.
\subsection{Orbit structure}

Let $S_n$ be the symmetric group on $N$. Given a subgroup $G\leq S_n$,
for any $\sigma \in G $ and $A \subseteq N$, an action of G on $2^N$ is defined by
\begin{align}
\sigma(A)= \{\sigma(i):i \in A \}.
\end{align}
Let $\mathcal{O}_G(A)=\{\sigma(A):\sigma \in G \}$ be an \emph{orbit} of $G$ that contains $A$, and
$\textfrak{O}_G=\{\mathcal{O}_G(A):A \subseteq N\}$ be the set of all orbits of $G$ on $2^N$. 
We call $\textfrak{O}_G$ the \emph{orbit structure} of $G$. Each orbit structure owns a partial order defined below.
%then \eqref{eq:example1} can be rewritten as 

\begin{definition}
For $\mc{O}_1,\mc{O}_2  \in \textfrak{O}_G$, $\mc{O}_1 \leq \mc{O}_2$ if for any $A \in \mc{O}_1$, 
there exists $B \in \mc{O}_2$ such that $A \subseteq B$.
\end{definition}

The partial order among subgroups of a symmetric group yields a partial order among the orbit structures they induced.
\begin{definition}
	\label{def2}
  $\textfrak{O}_{G_2}$ is a refinement of $\textfrak{O}_{G_1}$
  if for $G_1, G_2 \leq S_n$, $\textfrak{O}_{G_2} \leq \textfrak{O}_{G_1}$.
\end{definition}

\begin{example}
	Let $G_1 = S_2\mathrm{wr}_3S_3=\langle(123456),(16)(34)\rangle$\footnote{$S_2\mathrm{wr}_3S_3$ is the wreath product of $S_2$ and $S_3$, see the definition of wreath product in \textnormal{\cite[Section 2.6]{dixon1996permutation}},}. 
	$G_2=\langle(123)(456),(1563)(24)\rangle$. The Hasse diagram of the orbit structures of $G_1$ and $G_2$ are depicted in Figs. \ref{subfig:c} and \ref{subfig:e} respectively.
    By Definition \ref{def2}, we can see that $\textfrak{O}_{\langle(123)(456),(1563)(24)\rangle} \leq \textfrak{O}_{S_2\mathrm{wr}_3S_3}$.

%\begin{figure}[h]
%  \centering
%    \label{example}
%    \begin{subfigure}{0.45\linewidth}
%    \includegraphics[width=\linewidth]{orbit3}
%    \caption{$S_2\mathrm{wr}_3S_3$}
%    \label{subfig:a}
%  \end{subfigure}
%  \hspace{0.05\linewidth}
%  \begin{subfigure}{0.45\linewidth}
%    \includegraphics[width=\linewidth]{orbit5}
%    \caption{$\langle(123)(456),(1563)(24)\rangle$}
%    \label{subfig:b}
%  \end{subfigure}
%  \caption{Orbit structures}
% \end{figure}
\begin{figure}[h]
	\centering
	\begin{subfigure}{0.45\linewidth}
		\includegraphics[width=\linewidth]{orbit3}
		\caption{$\textfrak{O}_{S_2\mathrm{wr}_3S_3}$}
		\label{subfig:c}
	\end{subfigure}
	\begin{subfigure}{0.45\linewidth}
		\includegraphics[width=\linewidth]{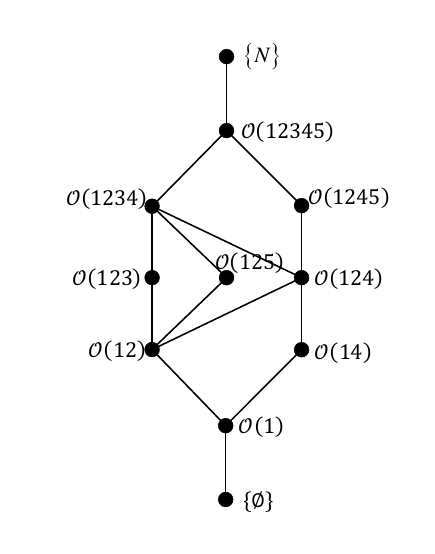}
		\caption{$\textfrak{O}_{\langle(123)(456),(1563)(24)\rangle}$}
		\label{subfig:e}
	\end{subfigure}
	\caption{Orbit structures}
    \label{orbit}
\end{figure}
\end{example}

\subsection{Symmetries in the entropy space}

The group action of a permutation group $G \leq S_n$ on $2^N$ induces a group action on $\mc{H}_n$, that is
%In this paper, we study the part of the entropy region  $\overline{\Gamma_n^*}$
%consisting of those almost entropic vectors which obey a specified collection of symmetries. 
%As a collection of symmetries must admit a group structure, we focus on symmetries arising from permuting random variables, 
%and consider any $G\leq S_n$, defining the associated group action on $\mc{H}_n$ by
\begin{align}
\sigma(\mb{h})(A)=\mb{h}(\sigma(A))
\end{align}
for any $\sigma \in G$ and $\mb{h} \in \mc{H}_n$. Let $\mathrm{fix}_G$ be the subspace of $\mc{H}_n$ fixed by the action, i.e.,
%\begin{align}
%\mathrm{fix}_G=\{\mb{h}\in \mc{H}_n:\mb{h}(A)=\mb{h}(B)\ if\  \exists\ \sigma \in G\ s.t.\ \sigma(A)=B\}\notag \label{eq:example1}\\
%\end{align}
\begin{align}
	\label{a6}
\mathrm{fix}_G=\{\mb{h}\in \mc{H}_n:\mb{h}(A)=\mb{h}(B)\ \text{if} \  A, B \in \mc{O}, \mc{O} \in \textfrak{O}_G\}.
\end{align}

It is defined in \cite{partition} that $G$-symmetric entropy region
\begin{align}
\Psi_G^* = \Gamma_n^* \cap \mathrm{fix}_G
\end{align}
and its outer bound, the $G$-symmetric polymatroidal region
\begin{align}
\Psi_G = \Gamma_n \cap \mathrm{fix}_G.
\end{align}
%Furthermore, by \cite[Thm.4]{partition} we know that
%\begin{align}
%\overline{\Psi_G^*} = \overline{\Gamma_n^*} \cap \mathrm{fix}_G.
%\end{align}

%Since the main purpose of this paper is to determine weather $\overline{\Psi_G^*}=\Psi_G$
%for arbitrary random variable permutation group G,
%while multiple groups can generate the same orbit structure in the power set,
%it is critical to study the structure of $\textfrak{O}_G$.
For self-containment, in the following, we list the three theorems about the relations between $G$-symmetric entropy regions and their outer bound $G$-symmetric polymatroidal regions in \cite{partition} and \cite{Symmetries16ITW}, which will facilitate our proofs in the rest of this paper.

\begin{theorem}\textnormal{\cite[Thm.1]{partition}}
	\label{lem1}
	Let $p=\{N_1, N_2, \cdots , N_t\}$ be a partition of $N$, that is, ${N_i}$, $i=1,2,\cdots ,t$ are disjoint and $\bigcup\limits_{i=1}^{t} N_i=N$. Let $G_p=S_{N_1}\times S_{N_2}\times \cdots \times S_{N_t}$, where $S_{N_i}$ are symmetric groups on ${N_i}$. For $\lvert N \rvert \geq 4$, 
	\begin{align}
		\overline{\Psi_{G_p}^*}=\Psi_{G_p}\notag
	\end{align}
	if and only if $p=\{N\}$ or $\{\{i\},N\textbackslash\{i\}\}$ for some $i \in N$.
\end{theorem}

%This Theorem, which considers special case of $(n_1,\cdots,n_t)$-partition groups
%$G_p=S_{N_1}\times S_{N_2}\times \cdots \times S_{N_t}$,
%states that the closure of $G_p$-symmetric entropy region can be 
%fully characterized by Shannon-type inequalities if and only if $p$ is the 1-partition 
%or a 2-partition with one of its blocks being a singleton.

 \begin{theorem}\textnormal{\cite[Thm.2]{Symmetries16ITW}}
\label{lem2}
For $G_1, G_2 \leq S_n$ with $\textfrak{O}_{G_1} \leq \textfrak{O}_{G_2}$,

\begin{enumerate}
  \item[1)] if $\overline{\Psi_{G_1}^*}=\Psi_{G_1} $, then $\overline{\Psi_{G_2}^*}=\Psi_{G_2} $;
  \item[2)] if $\overline{\Psi_{G_2}^*}\subsetneq \Psi_{G_2} $, then $\overline{\Psi_{G_1}^*}\subsetneq \Psi_{G_1} $.
\end{enumerate}

\end{theorem}

%This Theorem, which shows that Shannon sufficiency is implied
%under refinement of the associated partition of the power set,
%and hence to supergroups of a group, while insufficiency of
%Shanon-type inequalities is implied to subgroups.
%
%
%Futhermore, beyond those covered by Theorem \ref{lem2}, the following Theorem \ref{lem3}
%shows that non-Shannon-type inequalities are necessary under the action of some specific groups 
%$C_n$, $D_n$, $S_1\times C_{n-1}$ and $S_1\times D_{n-1}$.
 
In the following theorem, $C_n$ and $D_n$ are cyclic group and dihedral group  of degree $n$, respectively. 
\begin{theorem}\textnormal{\cite[Thm.5]{Symmetries16ITW}}
\label{lem3}
For $n \geq 6$, $\overline{\Psi_{C_n}^*}\subsetneq \Psi_{C_n} $, 
$\overline{\Psi_{D_n}^*}\subsetneq \Psi_{D_n} $,
$\overline{\Psi_{S_1\times C_{n-1}}^*}\subsetneq \Psi_{S_1\times C_{n-1}} $,
$\overline{\Psi_{S_1\times D_{n-1}}^*}\subsetneq \Psi_{S_1\times D_{n-1}} $.
\end{theorem}
%In this theorem, $C_n$ and $D_n$ are cyclic group and dihedral group  of degree $n$, respectively.
%This Theorem, together with Thm.s \ref{lem1} and \ref{lem2},
%give us many useful results for determining whether $\overline{\Psi_G^*}=\Psi_G$.

\section{Symmetric entropy functions of degree six and seven}
\label{results}

By definition, $\Psi_G^*$ and $\Psi_G$ are uniquely determined by the orbit structures $\textfrak{O}_{G}$. Moreover, if two permutation groups are in the same conjugacy class, their orbit structures are isomorphic up to the permutations of the indices in $N$. Therefore, permutation groups can be considered equivalent if their orbit structures are isomorphic. To determine whether $\overline{\Psi_G^*}=\Psi_G$, we only need to consider one representative in each equivalence class. For  $G \leq S_n$, $n = 6$, $7$, by $\mathsf{GAP}$\cite{GAP}, we see that there exsit $56$ and $96$ conjugacy classes, respectively. According to their orbit structure, they can be further classified into $35$ and $51$ equivalence classes. By Definition \ref{def2}, these equivalence classes form two partialy ordered sets $P_6$ and $P_7$, whose Hasse diagrams are depicted in Figures \ref{fig4}
and \ref{fig5} in Appendix B, respectively.

By Theorem \ref{lem2}, to determine the Shannon tightness of $\overline{\Psi_G^*}$, we only need to consider the antichains of all minimal equivalence classes in $P_6$ and $P_7$ such that  $\overline{\Psi_G^*}=\Psi_G$ with $\textfrak{O}_{G}$ in such an equivalence class, and maximal equivalence classes in $P_6$ and $P_7$ such that $\overline{\Psi_G^*}\subsetneq \Psi_G $ with $\textfrak{O}_{G}$ in such an equivalence class. We call these equivalence classes critical and they are characterized in Theorems \ref{Thm4} and \ref{Thm5}, and depicted in Figures \ref{fig2} and \ref{fig3}, respectively.

\subsection{Symmetric entropy functions of $6$ random variables}
%\begin{lemma}\textnormal{\cite[Prop. 6.4.10]{oxley2006matroid}}
%	\label{lem4}
%	If M is a matroid with fewer than eight elements, Then M is representable. 
%\end{lemma}

\begin{theorem}\label{Thm4}
For any $G\leq S_6$, $\overline{\Psi_G^*}=\Psi_G$ if and only if $G =$ 
%\begin{itemize}
%	\item $S_6$, $A_6$, $ \mathrm{PGL}_2(5)$\footnote{We adopt the notations of permutation groups in \cite{dixon1996permutation} in this paper. The generating element representation of these groups can be found in Figs. \ref{fig2} and \ref{fig3}.},
%	\item $\mathrm{PSL}_2(5),$
%	\item $S_1 \times S_5$, $ S_1 \times A_5$, $ S_1 \times \mathrm{AGL}_1(5)$.
%\end{itemize}
\begin{itemize}
	\item $S_6:\langle(123456),(12)\rangle$, \\  $A_6:\langle(12345),(456)\rangle$, \\  $ \mathrm{PGL}_2(5):\langle(123456),(16)(23)\rangle$\footnote{We adopt the notations of permutation groups in \cite{dixon1996permutation} in this paper. The generating element representation of these groups can be found in Figs. \ref{fig2} and \ref{fig3}.},
	\item $\mathrm{PSL}_2(5):\langle(123)(456),(14623)\rangle,$
	\item $S_1 \times S_5:\langle(23456),(23)\rangle$, \\  $ S_1 \times A_5:\langle(23456),(234)\rangle$,  \\ $ S_1 \times \mathrm{AGL}_1(5):\langle(23456),(2546)\rangle$.
\end{itemize}
\end{theorem}

\begin{proof}
%$n = 6$ : %The only remaining cases not already handled by Thm.s 1, 2 and 5 is that
%%$PSL_2(5):\langle(123)(456),(14623)\rangle$, $S_3wr_2C_2:\langle(12)(34)(56),(14)(235)\rangle$, 
%%$S_2wr_3S_3:\langle(12)(34)(56),(16)(34)\rangle$ and its subgroups
%%$\langle(123)(456),(1563)(24)\rangle$, $\langle(123456),(14)(36)\rangle$ and $\langle(123)(456),(16)(35)\rangle$.
%To best leverage subgroup and supergroup implications set out by \ref{lem1},
%the study should focus on maximal or minimal subgroup. The main subgroups 
%that need to be considered, along with their corresponding results, are depicted in  Fig. \ref{fig2}.
%
%The only remaining cases not already handled by Thm.s \ref{lem1}, \ref{lem2} and \ref{lem3}
%are the symmetric entropic region induced by $\mathrm{PSL}_2(5)$, $S_3\mathrm{wr}_2C_2$,
% $S_2\mathrm{wr}_3S_3$, $C_2\mathrm{wr}_3C_3$, $\langle (123)(456),(1563)(24) \rangle$ 
%and $A_6 \cap S_1 \times S_1 \times S_4$.
%Since insufficiency of Shanon-type inequalities is implied to subgroups 
%by Theorem \ref{lem1}, it suffices to prove 
%$\overline{\Psi_{\mathrm{PSL}_2(5)}^*}=\Psi_{\mathrm{PSL}_2(5)}$, $\overline{\Psi_{S_3\mathrm{wr}_2C_2}^*} \neq \Psi_{S_3\mathrm{wr}_2C_2}$
%and $\overline{\Psi_{S_2\mathrm{wr}_3S_3}^*} \neq \Psi_{S_2\mathrm{wr}_3S_3}$,
%as this information handle the other three cases via implication.
To prove the theorem, by Theorem \ref{Lem2}, we only need to determine the minimal equivalence classes of subgroups of $S_6$ in $P_6$ such that $\overline{\Psi_G^*}=\Psi_G$ and the maximal equivalence classes of subgroups of $S_6$ in $P_6$ such that $\overline{\Psi_G^*} \subsetneq \Psi_G$. According to the Hasse diagram of $P_6$ depicted in Figure \ref{fig4}, we will prove in the following that these two families of equivalence classes of subgroups are
\begin{itemize}
	\item $\mathrm{PSL}_2(5):\langle (123)(456),(14623) \rangle$,
	\item $S_1 \times S_5:\langle (23456),(23) \rangle$,  \\ $ S_1 \times A_5:\langle (23456),(234) \rangle$ , \\  $S_1 \times \mathrm{AGL}_1(5):\langle (23456),(2546) \rangle$,
\end{itemize}
and
\begin{itemize}
	\item  $S_2 \times S_4:\langle (12)(3456),(34) \rangle $, \\  $S_2 \times A_4: \langle (12)(345)$, $(34)(56) \rangle$, \\  $ \langle (12)(3456),(12)(36) \rangle $,
	\item  $S_3\times S_3:\langle (123)(56),(23)(456) \rangle $, \\  $S_3 \times A_3: \langle (123),(12),(456) \rangle$, \\  $A_3 \times A_3: \langle (123)$, $(456)\rangle, \\ \langle (123)(456), (13)(46), (23)(46) \rangle $,
	\item  $S_1 \times C_5:\langle (23456) \rangle $, \\  $S_1\times D_5:\langle (23456),(36)(45) \rangle $,
	\item $S_2\mathrm{wr}_3S_3:\langle (123456),(16)(34)\rangle$, \\ $\langle (123)(456),(13)(24)(56) \rangle$,
	\item $S_3\mathrm{wr}_2C_2:\langle (12)(34)(56),(14)(235) \rangle$,  \\ $A_6 \cap S_3\mathrm{wr}_2C_2:\langle (12)(3456),(16)(23) \rangle$ , \\  $3^2 \cdot 2^2:\langle (123456),(15)(26) \rangle$, \\  $C_3\mathrm{wr}_2C_2:\langle (123456),(153)(246) \rangle$.
\end{itemize}
%\begin{itemize}
%	\item $\mathrm{PSL}_2(5)$,
%	\item $S_1 \times S_5$, $ S_1 \times A_5$ , $S_1 \times \mathrm{AGL}_1(5)$,
%\end{itemize}
%and
%\begin{itemize}
%	\item  $S_2 \times S_4 $, $S_2 \times A_4$, $ \langle (12)(3456),(12)(36) \rangle $,
%	\item  $S_3\times S_3 $, $S_3 \times A_3$, $A_3 \times A_3$, $\langle (123)(456)$,$(13)(46)$, $ (23)(46) \rangle $,
%	\item  $S_1 \times C_5 $, $S_1\times D_5 $,
%	\item $S_2\mathrm{wr}_3S_3$, $\langle (123)(456),(13)(24)(56) \rangle$,
%	\item $S_3\mathrm{wr}_2C_2$, $A_6 \cap S_3\mathrm{wr}_2C_2$, $3^2 \cdot 2^2$, $C_3\mathrm{wr}_2C_2$,
%\end{itemize}
respectively. The orbit structures of them are depicted in Appendix A.
%\begin{figure}[h]
%	\centering
%	\label{orbit}
%	\begin{subfigure}{0.45\linewidth}
%		\includegraphics[width=\linewidth]{orbit1}
%		\caption{$\mathrm{PSL}_2(5)$}
%		\label{subfig:a}
%	\end{subfigure}
%	\hspace{0.05\linewidth}
%	\begin{subfigure}{0.43\linewidth}
%		\includegraphics[width=\linewidth]{orbit2}
%		\caption{$S_1 \times S_5$}
%		\label{subfig:b}
%	\end{subfigure}
%    \begin{subfigure}{0.45\linewidth}
%    	\includegraphics[width=\linewidth]{orbit3}
%    	\caption{$S_2\mathrm{wr}_3S_3$}
%    	\label{subfig:c}
%    \end{subfigure}
%     \begin{subfigure}{0.43\linewidth}
%     	\includegraphics[width=\linewidth]{orbit4}
%     	\caption{$S_3\mathrm{wr}_2C_2$}
%     	\label{subfig:d}
%     \end{subfigure}
%  \begin{subfigure}{0.45\linewidth}
%	\includegraphics[width=\linewidth]{orbit5}
%	\caption{$\langle(123)(456),(1563)(24)\rangle$}
%	\label{subfig:e}
%  \end{subfigure}
%	\caption{Orbit structures}
%	
%\end{figure}
\begin{figure*}[htp]
	\centering
	\includegraphics[width=1.0\linewidth,height=0.5\linewidth]{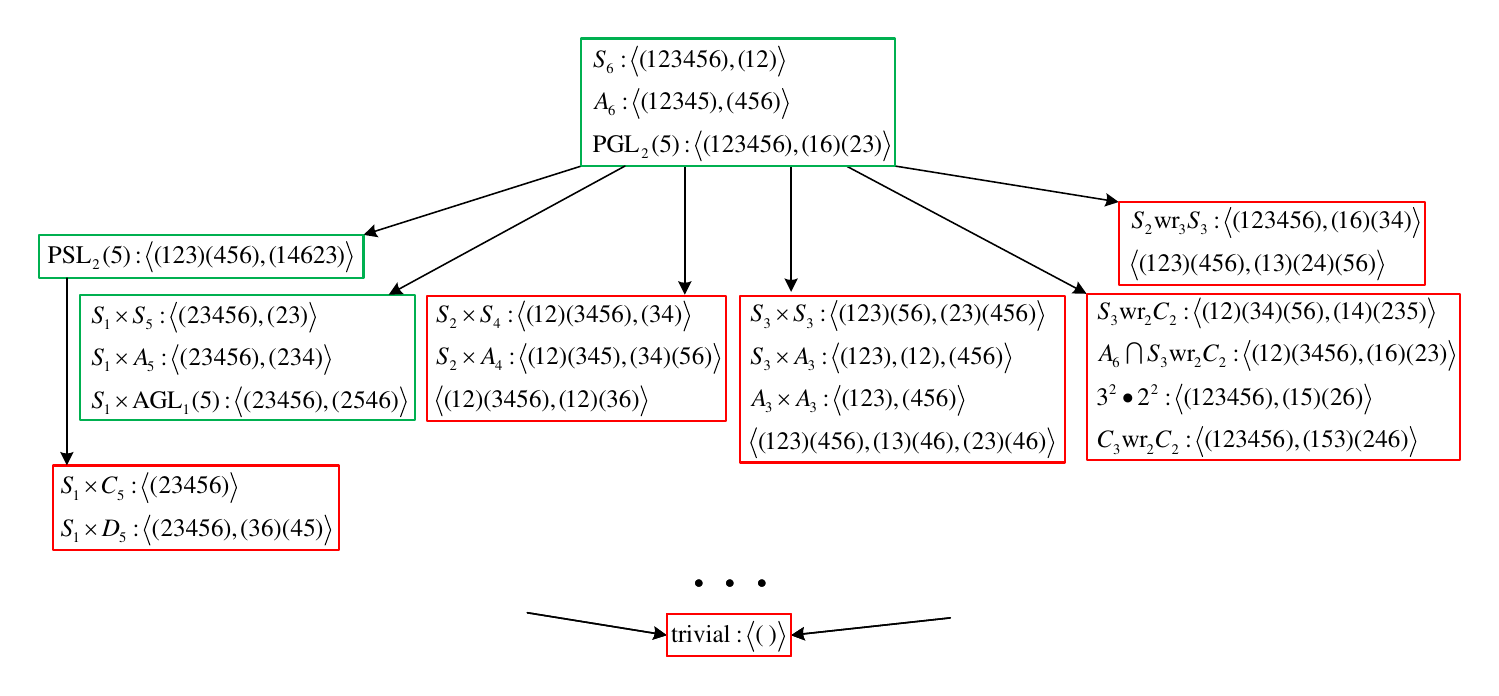}
	\captionsetup{font=small}
	\caption{The Hasse diagram of critical equivalence classes of $P_6$. 
		The "green" nodes indicate $\overline{\Psi_G^*} = \Psi_G $, 
		and "red" nodes indicate $\overline{\Psi_G^*}\subsetneq \Psi_G $.}
	\label{fig2}
\end{figure*}
\begin{spacing}{1.2}
\noindent\underline{$\overline{\Psi_G^*} = \Psi_G$ for $G=\mathrm{PSL}_2(5)$:} 
\end{spacing}
Reducing the redundancies in \eqref{a1}-\eqref{a3}, inequalities determining $\Gamma_n$ are in the following elemental form \textnormal{\cite[Section 14.1]{yeung2008information}} 
\begin{enumerate}
	\item $\mathbf{h}(N) \geq \mathbf{h}(N\textbackslash\{i\})$,  $i \in N$
	\item $\mathbf{h}(i \cup K) + \mathbf{h}(j \cup K)\geq \mathbf{h}(K)+\mathbf{h}(\{i,j\}\cup K)$,  distinct  $i,j \in N$, $K\subseteq N\textbackslash \{i,j\}  $,
\end{enumerate}
with each corresponds to a facet of $\Gamma_n$. According to the orbit structure of $G$, $\overline{\Psi_G^*} $ and $ \Psi_G$ live in a subspace $\mathbb{R}^{\textfrak{O}_G}$ of $\mc{H}_n$, which is indexed by $(s_{\mathcal{O}}:\mathcal{O} \in \textfrak{O}_G)$. Imposing the symmetric constraints in \eqref{a6}, we have the elemental form of the inequalities bounding $\Psi_G$, i.e., the H-representation \cite{ziegler2012lectures} of $\Psi_G$. Input these inequalities into  $\mathsf{lrs}$\cite{Lrs}, we have the V-representation of $\Psi_G$, i.e., the extreme rays of $\Psi_G$. The H-representation and V-representation of $\Psi_G$ can be found in Appendix A.

To prove $\overline{\Psi_G^*} = \Psi_G $, it is sufficient to show that all the extreme rays of $\Psi_G$ are almost entropic. Among the $8$ extreme rays of $\Psi_G$, $6$ are those containing uniform matroids $U_{i,n}$, $i=1,2,\cdots,6$, which are representative and so almost entropic. The remaining $2$ extreme rays are symmetric by index permutation, so we only need to show one of them is almost entropic. Consider the one containing polymatroid 
%Since any linear inequality valid for Shannon entropy is valid for ranks(dimensions) in any linear space over any finite field or over $\mathbb{R}$\cite{hammer2000inequalities}, it can be proved by Lemma \ref{lem1}.
%We can obtain all the 8 extreme rays of $\Psi_{\mathrm{PSL}_2(5)}$ by $\mathsf{lrs}$\cite{Lrs},
%among which one polymatroid given by the vector $\mathbf{h_1}$:
$$\mathbf{h_1}(A)=\left\{
\begin{aligned}
	2 \quad& \text{if}\ |A| = 1,\\
	4 \quad& \text{if}\ |A| = 2 ,\\
	5 \quad& \text{if}\ A \in \mathcal{O}(123),\\
	6 \quad& \text{otherwise}.
	%6 \quad& \text{if}\ A \in \mathcal{O}(126)\ \text{or}\ |A| \geq 4.
\end{aligned}
\right.
$$
It can be shown that it is multi-linear representable over $\mathrm{GF}(11)$. The following is a $6 \times 12$ matrix over $\mathrm{GF}(11)$, and each $i \in N$ labels a $6 \times 2$ submatrix. It can be checked that it represents $\mathbf{h_1}$, that is, for each subset $A \subseteq N = \{1,2,\cdots,6\}$, the rank of the $6 \times 2|A|$ submatrix whose columns are labeled by $A$ is equal to $\mathbf{h_1}(A)$.
%\begin{center}
%$2$ $2$ $2$ $2$ $2$ $2$ $4$ $4$ $4$ $4$ $4$ $4$ $4$ $4$ $4$ $4$ $4$ $4$ $4$ $4$ $4$\\
%$6$ $5$ $6$ $5$ $6$ $5$ $5$ $5$ $6$ $6$ $5$ $5$ $6$ $6$ $6$ $5$ $6$ $5$ $6$ $5$ $6$\\
%$6$ $6$ $6$ $6$ $6$ $6$ $6$ $6$ $6$ $6$ $6$ $6$ $6$ $6$ $6$ $6$ $6$ $6$ $6$ $6$ $6$\\
%\end{center}

%
%is representable over $\mathrm{GF}(11)$, and we can give the representation in Table \ref{tab1} for $\mb{h_1}$.
%Since a representable polymatroid is always almost entropic,
%we can conclude that $\mathbf{h_1} \in \overline{\Psi_{\mathrm{PSL}_2(5)}^*}$.
%Another extreme ray is isomophisim with $\mathbf{h_1}$ by a permutation $(3456)$ in the power set.
%Besides, the remaining six polymatroids are all uniform matroids, and thus entropic, which implies 
%$\overline{\Psi_{\mathrm{PSL}_2(5)}^*}=\Psi_{\mathrm{PSL}_2(5)}$.

%\[
%\begin{bNiceArray}{ccc:cc:cc:cc:cc:ccc}[first-row]
%	&1& 1 &2& 3&3& 3&4& 5&5& 5&6& 6&\\
%	&1& 0 &0& 0&0& 1&0& 0 &0& 1&4& 6 & \\
%	&0& 1 &0& 0&1& 0&0& 0 &0& 1&2& 3 &\\
%	&0& 0 &1& 0&0& 1&0& 0 &0& 1&7& 2 & \\
%	&0& 0 &0& 1&0& 1&0& 0 &1& 0&9& 1 & \\
%	&0& 0 &0& 0&0& 0&1& 0 &1& 0&9& 9 & \\
%	&0& 0 &0& 0&1& 0&0& 1 &0& 0&0& 1 & \\
%\end{bNiceArray}
%\] 

\[
\begin{bNiceArray}{@{}cccccccc@{}}[first-row]
\label{matrix1}
	&1  &2 &3 &4 &5 &6 &\\
	&1\  0 &0\  0&0\  1&0\  0 &0\  1&4\  6 & \\
	&0\  1 &0\  0&1\  0&0\  0 &0\  1&2\  3 & \\
	&0\  0 &1\  0&0\  1&0\  0 &0\  1&7\  2 & \\
	&0\  0 &0\  1&0\  1&0\  0 &1\  0&9\  1 & \\
	&0\  0 &0\  0&0\  0&1\  0 &1\  0&9\  9 & \\
	&0\  0 &0\  0&1\  0&0\  1 &0\  0&0\  1 & \\
	\CodeAfter
	\tikz \draw [dashed, line width=0.6pt, black] (1-|3) -- (7-|3) ;
	\tikz \draw [dashed, line width=0.6pt, black] (1-|4) -- (7-|4) ;
	\tikz \draw [dashed, line width=0.6pt, black] (1-|5) -- (7-|5) ;
	\tikz \draw [dashed, line width=0.6pt, black] (1-|6) -- (7-|6) ;
	\tikz \draw [dashed, line width=0.6pt, black] (1-|7) -- (7-|7) ;
\end{bNiceArray}
\] 
%\begin{table}[htp]
%\caption{The representation of $\mathbf{h_1}$}
%\label{tab1}
%\renewcommand{\arraystretch}{1.5}
%\centering
%\begin{tabular}{|c|c|c|c|c|c|}
%\hline
%$X_1$ & $X_2$ & $X_3$ &$X_4$ & $X_5$ & $X_6$\\
%\hline
%1 0 & 0 0 & 0 1 & 0 0 & 0 1 & 4 6\\
%0 1 & 0 0 & 1 0 & 0 0 & 0 1 & 2 3\\
%0 0 & 1 0 & 0 1 & 0 0 & 0 1 & 7 2\\
%0 0 & 0 1 & 0 1 & 0 0 & 1 0 & 9 1\\
%0 0 & 0 0 & 0 0 & 1 0 & 1 0 & 9 9\\
%0 0 & 0 0 & 1 0 & 0 1 & 0 0 & 0 1\\
%\hline
%\end{tabular}
%\end{table}
\begin{spacing}{1.1}
\noindent\underline{$\overline{\Psi_G^*} =\Psi_G$ for $G=S_1\times S_5$:} 
\end{spacing}
By Theorem \ref{lem1}, $\overline{\Psi_G^*} =\Psi_G$ for $G=S_1\times S_5$. As $S_1\times A_5$ and $S_1\times \mathrm{AGL}_1(5)$ have the same orbit structures with $S_1\times S_5$, $\overline{\Psi_G^*} =\Psi_G$ for $G=S_1\times A_5$, $\mathrm{AGL}_1(5)$.

\begin{spacing}{1.2}
\noindent\underline{$\overline{\Psi_G^*} \neq \Psi_G$ for $G=S_2\times S_4$, $S_3\times S_3$, $S_1\times C_5$:} 
\end{spacing}
It is immediately implied  by Theorems \ref{lem1} and \ref{lem3}.
%By Theorem \ref{lem1}, $\overline{\Psi_G^*} \neq \Psi_G$ for $G=S_2\times S_4$, $S_3\times S_3$. Since permutation groups in one equivalence class have the same orbit structures, the results is proved.
%\begin{spacing}{1.2}
%\noindent\underline{$\overline{\Psi_G^*} \neq \Psi_G$ for $G=S_1\times C_5$:} 
%\end{spacing}
%It is immediately implied  by Theorem \ref{lem3}.

\begin{spacing}{1.1}
\noindent\underline{$\overline{\Psi_G^*} \neq \Psi_G$ for $G=S_3\mathrm{wr}_2C_2$:} 
\end{spacing}
By a similar treatment for $\mathrm{PSL}_2(5)$, we obtain the H-representation and V-representation of $\Psi_G$, $G=S_3\mathrm{wr}_2C_2$, which can be find in Appendix A. Among the extreme ray of $\Psi_G$, there is one containing the following polymatroid
$$\mathbf{h_2}(A)=\left\{
\begin{aligned}
	2 \quad& \text{if}\ |A| = 1,\\
	4 \quad& \text{if}\ |A| = 2 ,\\
	5 \quad& \text{if}\ A \in \mathcal{O}(124),\\
	6 \quad& \text{otherwise}.
	%6 \quad& \text{if}\ A \in \mathcal{O}(126)\ \text{or}\ |A| \geq 4.
\end{aligned}
\right.
$$
%\begin{center}
%$2$ $2$ $2$ $2$ $2$ $2$ $4$ $4$ $4$ $4$ $4$ $4$ $4$ $4$ $4$ $4$ $4$ $4$ $4$ $4$ $4$\\
%$5$ $5$ $5$ $5$ $5$ $5$ $5$ $5$ $6$ $5$ $5$ $6$ $5$ $5$ $5$ $5$ $5$ $5$ $5$ $5$ $6$\\
%$6$ $6$ $6$ $6$ $6$ $6$ $6$ $6$ $6$ $6$ $6$ $6$ $6$ $6$ $6$ $6$ $6$ $6$ $6$ $6$ $6$.\\
%\end{center}
By contracting $\{6\}$ from $\mathbf{h_2}$, then restricting it on $\{1,2,3,4\}$, we have
$$ \widetilde{\mathbf{h_2}}(A)=\left\{
\begin{aligned}
2 \quad& \text{if}\ |A| = 1, \\
3 \quad& \text{if}\ |A| = 2 \ \text{and} \ |A| \neq \{1,4\},\\
4 \quad& \text{if}\ |A| = \{1,4\}\ \text{or}\ |A| \geq 3,
\end{aligned}
\right.
$$
which violates Zhang-Yeung inequality\cite{Zhang1998}, and is therefore non-entropic. Thus $\overline{\Psi_{G}^*} \neq \Psi_{G}$.

\noindent\underline{$\overline{\Psi_G^*} \neq \Psi_G$ for $G=S_2\mathrm{wr}_3S_3$:} 

Among the extreme rays of $\Psi_G$, one finds the polymatroid given by
$$\mathbf{h_3}(A)=\left\{
\begin{aligned}
	2 \quad& \text{if}\ |A| = 1,\\
	3 \quad& \text{if}\ A \in \mathcal{O}(12),\\
	4 \quad& \text{otherwise}.
	%6 \quad& \text{if}\ A \in \mathcal{O}(126)\ \text{or}\ |A| \geq 4.
\end{aligned}
\right.
$$
%\begin{center}
%$2$ $2$ $2$ $2$ $2$ $2$ $3$ $3$ $4$ $3$ $3$ $3$ $3$ $4$ $3$ $3$ $3$ $4$ $3$ $3$ $3$\\
%$4$ $4$ $4$ $4$ $4$ $4$ $4$ $4$ $4$ $4$ $4$ $4$ $4$ $4$ $4$ $4$ $4$ $4$ $4$ $4$ $4$\\
%$4$ $4$ $4$ $4$ $4$ $4$ $4$ $4$ $4$ $4$ $4$ $4$ $4$ $4$ $4$ $4$ $4$ $4$ $4$ $4$ $4$.\\ 
%\end{center}
By restricting the polymatroid on \{2,3,4,5\}, we have
$$ \widetilde{\mathbf{h_3}}(A)=\left\{
\begin{aligned}
2 \quad& \text{if}\ |A| = 1, \\
3 \quad& \text{if}\ |A| = 2 \ \text{and} \ A \neq \{2,5\},\\
4 \quad& \text{if}\ A = \{2,5\}\ \text{or}\ |A| \geq 3,
\end{aligned}
\right.
$$
which violates Zhang-Yeung inequality, and is therefore non-entropic. Thus $\overline{\Psi_{G}^*} \neq \Psi_{G}$.
\end{proof}

\subsection{Symmetric entropy functions of $7$ random variables}

%\begin{lemma}
%	\label{Lem2}
%	$\overline{\Psi_{G_1}^*} \neq \Psi_{G_1}  \Rightarrow 
%	\overline{\Psi_{G_1 \times G_2}^*} \neq \Psi_{G_1 \times G_2}$.
%\end{lemma}
%\begin{proof}
%	%	Let $N=N_1 \cup N_2$ and $\mc{H}_{n}=\mathbb{R}^{2^N}$. 
%	Let $G_1\leq S_{n_1}$ and $G_2\leq S_{n_2}$ such that $n=n_1+n_2$.
%	If $\overline{\Psi_{G_1}^*} \neq \Psi_{G_1}$, then there exsits $\mb{h_1} \in \Psi_{G_1}$
%	such that $\mb{h_1} \notin \overline{\Psi_{G_1}^*} = \overline{\Gamma_{n_1}^*} \cap \mr{fix}_{G_1}$.
%	Since $\mb{h_1} \in \Psi_{G_1} = \Gamma_{n_1} \cap \mr{fix}_{G_1}$, 
%	$\mb{h_1} \in \mr{fix}_{G_1} $ which implies $\mb{h_1} \notin \overline{\Gamma_{n_1}^*}$.
%	Without loss of generality, for any $\mb{h_2} \in \mr{fix}_{G_2} $, 
%	we assume the groud set of $\mb{h_1}$ and $\mb{h_2}$ are $N_1$ and $N_2$,
%	such that $N_1 \cup N_2 = N$.
%	Since $\mb{h_1}$ is a polymatroid, we expand it to $\mb{h_1}'$ with  
%	$n$ element by adding $n - n_1$ loops. Then we can check
%	
%	
%	\begin{align}
%		\mb{h_1}'(A) = \mb{h_1}(A \cap N_1),\quad \forall A \subseteq N \notag 
%	\end{align}
%	With the same method, we expand $\mb{h_2}$ to $\mb{h_2}'$
%	by adding $n - n_2$ loops. Let $\mb{h}=\mb{h_1'} +\mb{h_2'}$, then
%	\begin{align}
%		\mb{h}(A)=\mb{h_1}(A \cap N_1)+\mb{h_2}(A \cap N_2), \quad\forall A \subseteq N. \notag 
%	\end{align}
%	Since
%	$\mb{h} \in \mr{fix}_{G_1 \times G_2}$ but $\mb{h} \notin \overline{\Gamma_{n}^*}$,
%	we have $\mb{h} \notin \overline{\Psi_{G_1 \times G_2}^*}$. Therefore $\overline{\Psi_{G_1 \times G_2}^*} \neq \Psi_{G_1 \times G_2}$, which proves the Lemma.
%\end{proof}
\begin{lemma}
	\label{Lem2}
	For symmetric group $S_{N_1}$ and $S_{N_2}$, with $N_1$ and $N_2$ disjoint, let $G_1\leq S_{N_1}$ and $G_2\leq S_{N_2}$, then $\overline{\Psi_{G_1}^*} \neq \Psi_{G_1} $ implies  $
	\overline{\Psi_{G_1 \times G_2}^*} \neq \Psi_{G_1 \times G_2}$.
\end{lemma}
\begin{proof}
	%	Let $N=N_1 \cup N_2$ and $\mc{H}_{n}=\mathbb{R}^{2^N}$. 
	As $\overline{\Psi_{G_1}^*} \neq \Psi_{G_1}$, there exsits $\mb{h_1} \in \Psi_{G_1}$
	such that $\mb{h_1} \notin \overline{\Psi_{G_1}^*} = \overline{\Gamma_{n_1}^*} \cap \mr{fix}_{G_1}$.
	Since $\mb{h_1} \in \Psi_{G_1} = \Gamma_{n_1} \cap \mr{fix}_{G_1}$, 
	$\mb{h_1} \in \mr{fix}_{G_1} $ which implies $\mb{h_1} \notin \overline{\Gamma_{n_1}^*}$.
	%	Without loss of generality, for any $\mb{h_2} \in \mr{fix}_{G_2} $, 
	%	we assume the groud set of $\mb{h_1}$ and $\mb{h_2}$ are $N_1$ and $N_2$,
	%	such that $N_1 \cup N_2 = N$.
	By adding $|N_2|$ loops to  $\mb{h_1}$, we obtain $\mb{h_1}'$ on  $N\triangleq N_1\cup N_2$ with its rank function
	\begin{align}
		\mb{h_1}'(A) = \mb{h_1}(A \cap N_1),\quad \forall A \subseteq N. \notag 
	\end{align}
	For any $\mb{h_2}\in \Psi_{G_2}$, we obtain $\mb{h_2}'$ similarly. Let   $\mb{h}=\mb{h_1'} +\mb{h_2'}$, then   $\mb{h} \in \Psi_{G_1 \times G_2}$ but  $\mb{h} \notin \overline{\Psi_{G_1 \times G_2}^*}$. Therefore $\overline{\Psi_{G_1 \times G_2}^*} \neq \Psi_{G_1 \times G_2}$, which proves the lemma.
\end{proof}
\begin{theorem}\label{Thm5}
	For any $G\leq S_7$, $\overline{\Psi_G^*}=\Psi_G$ if and only if $G =$
%	\begin{itemize}
%		\item $S_7,\ A_7$,
%		\item $S_1 \times S_6,\ S_1 \times A_6,\ \langle (234567),(27)(34) \rangle$,
%		\item $\mathrm{AGL}_1(7)$.
%	\end{itemize}
	\begin{itemize}
	\item $S_7:\langle (1234567),(12) \rangle, \\ \ A_7:\langle (1234567),(123) \rangle$,
	\item $S_1 \times S_6:\langle (234567),(23) \rangle, \\ \ S_1 \times A_6:\langle (234567),$ $(567) \rangle, \\ \ \langle (234567),(27)(34) \rangle$,
	\item $\mathrm{AGL}_1(7):\langle (1234567),(163247) \rangle$.
\end{itemize}
\end{theorem}

\begin{proof}
	By the same argument in Theorem \ref{Thm4}, the equivalence classes of subgroups of $S_7$ that need to be considered, are the following two families
\begin{itemize}
	
	\item $S_1 \times S_6:\langle (234567),(23) \rangle$,  \\ $ S_1 \times A_6:\langle (234567),(567) \rangle$, \\  $\langle (234567),(27)(34) \rangle$,
	\item $\mathrm{AGL}_1(7):\langle (1234567),(163247) \rangle$,
\end{itemize}
and
\begin{itemize}
	\item $S_2 \times S_5:\langle(12),(34567),(34)\rangle $, \\  $S_2 \times A_5:\langle(34567)$,  $(12),(345)\rangle$, \\  $\langle(12)(345)(67),(12)(3765)\rangle$,  \\ $\langle(12)(34567)$,   $(12)(3651)\rangle$,
	\item $S_3 \times S_4:\langle(123)(47),(13)(4567)\rangle$, \\  $\langle(123)(45)(67),(13)(4567),(13)(47)\rangle$, \\ $
	\langle(123)(457),(13)(47)(56)\rangle$, \\ $\langle(123)(47),(123)(4567)\rangle$, \\ $\langle(123)(47)(56),(123)(457)\rangle
	$,
	\item $S_1 \times \mathrm{PSL}_2(5):\langle (234)(567),(25734) \rangle$,
	\item $S_1\times 3^2 \cdot 2^2$ $:\langle (234567),(26)(37) \rangle$, \\  $S_1\times S_3\mathrm{wr}_2C_2:\langle(23)(45)(67),(25)(346) \rangle$, \\  $ S_1\times A_6\cap$ $ S_3\mathrm{wr}_2C_2:$ $\langle (4567)(23),(27)(34) \rangle$,  \\ 
	$S_1 \times C_3\mathrm{wr}_2C_2:\langle (234567),(264)(357) \rangle$,
	\item $S_1\times S_2\mathrm{wr}_3S_3:\langle (234567),(27)(45) \rangle$, \\  $ 
	\langle (234)(567),(24)(35)(67) \rangle$,
	\item $\mathrm{PGL}_3(2):\langle (1234567),(1367)(45) \rangle$.
\end{itemize}
%\begin{itemize}
%	
%	\item $S_1 \times S_6$, $ S_1 \times A_6$, $\langle (234567),(27)(34) \rangle$,
%	\item $\mathrm{AGL}_1(7)$,
%\end{itemize}
%and
%\begin{itemize}
%	\item $S_2 \times S_5 $, $S_2 \times A_5$, $\langle(12)(345)(67),(12)(3765)\rangle$, $\langle(12)(34567)$, $(12)(3651)\rangle$,
%	\item $S_3 \times S_4$, $\langle(123)(45)(67),(13)(4567),(13)(47)\rangle$, $
%	\langle(123)$ $(457)$,$(13)(47)(56)\rangle$,$\langle(123)(47)$,$(123)(4567)\rangle$,$\langle(123)(47)$ $(56),(123)(457)\rangle
%	$,
%	\item $S_1 \times \mathrm{PSL}_2(5)$,
%	\item $S_1\times S_3\mathrm{wr}_2C_2$, $ S_1\times A_6\cap S_3\mathrm{wr}_2C_2$, $S_1\times 3^2 \cdot 2^2$,  	$\qquad \qquad S_1 \times C_3\mathrm{wr}_2C_2$,
%	\item $S_1\times S_2\mathrm{wr}_3S_3$, $ 
%	\langle (234)(567)$, $(24)(35)(67) \rangle$,
%	\item $\mathrm{PGL}_3(2)$.
%\end{itemize}

%$n = 7$ : It suffices to prove 
%$\overline{\Psi_{S_1 \times  S_3\mathrm{wr}_2C_2}^*} \neq \Psi_{S_1 \times  S_3\mathrm{wr}_2C_2}$ ,
%$\overline{\Psi_{S_1 \times  S_2\mathrm{wr}_3S_3}^*} \neq \Psi_{S_1 \times  S_2\mathrm{wr}_3S_3}$,
%$\overline{\Psi_{\mathrm{PGL}_3(2)}^*} \neq \Psi_{\mathrm{PGL}_3(2)}$ 
%and $\overline{\Psi_{\mathrm{AGL}_1(7)}^*}\stackrel{?}{=} \Psi_{\mathrm{AGL}_1(7)}$,
%as Thm.s \ref{lem1}, \ref{lem2} and \ref{lem3} together with this information handle
%the other cases via implication.

\begin{figure*}[htp]
	\centering
	%\vspace{-1cm}
	\includegraphics[width=1.0\linewidth,height=0.5\linewidth]{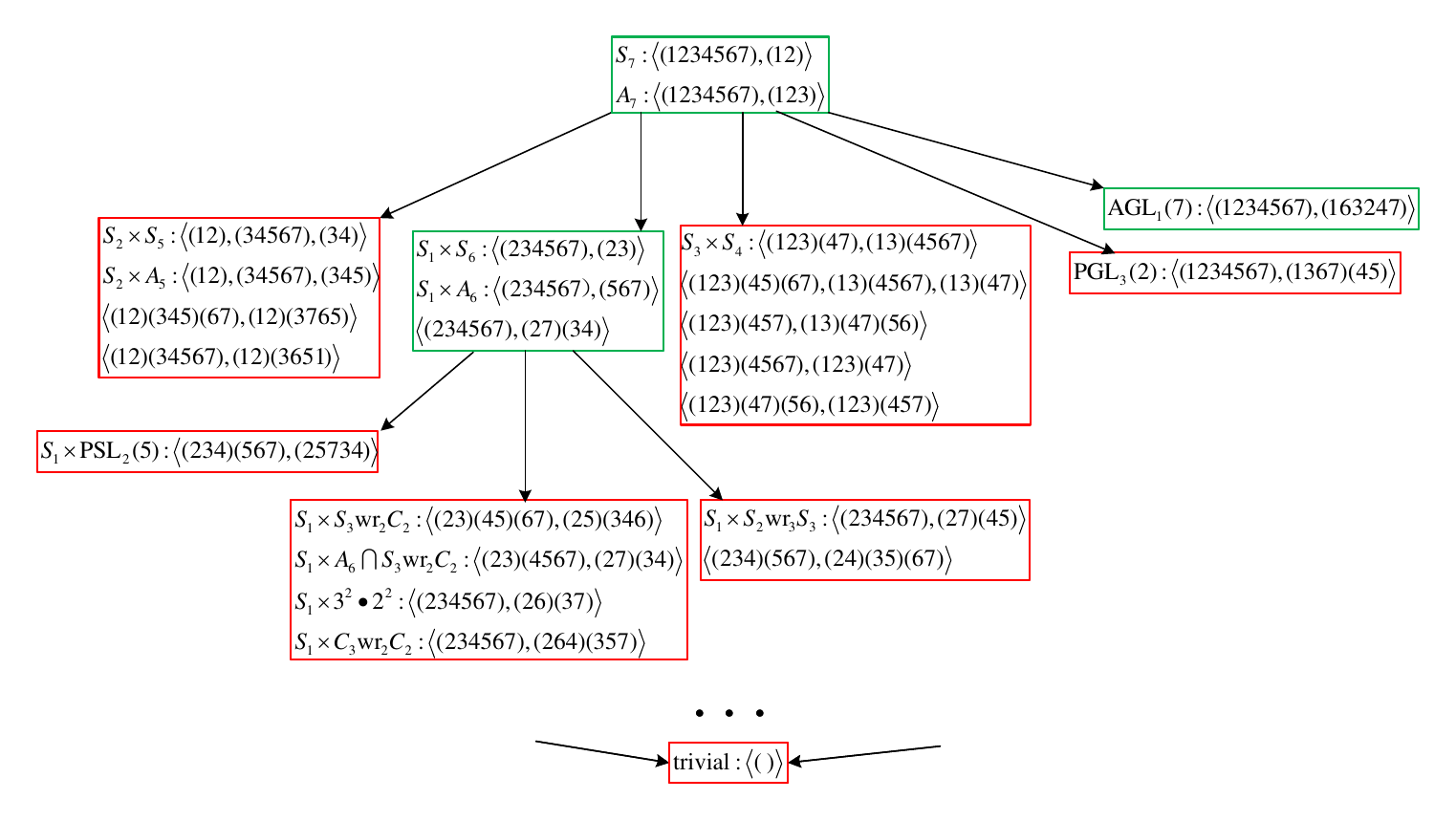}
	\captionsetup{font=small}
	\caption{The Hasse diagram of critical part of $P_7$ }
	\label{fig3}
\end{figure*}
\begin{spacing}{1.2}
\noindent\underline{$\overline{\Psi_G^*} =\Psi_G$ for $G=S_1\times S_6$:} 
\end{spacing}
It is immediately implied  by Theorem \ref{lem1}.
%By theorem \ref{lem1}, $\overline{\Psi_G^*} =\Psi_G$ for $G=S_1\times S_6$. As $S_1\times A_6$ and $\langle (234567),(27)(34) \rangle$ have the same orbit structures with $S_1\times S_6$, $\overline{\Psi_G^*} =\Psi_G$ for $G=S_1\times A_6$, $\langle (234567),(27)(34) \rangle$.
\begin{spacing}{1.2}
\noindent\underline{$\overline{\Psi_G^*} =\Psi_G$ for $G=\mathrm{AGL}_1(7)$:} 
\end{spacing}
Through $\mathsf{lrs}$ we can obtain all the 13 extreme rays of $\Psi_G$, 
among which 7 are uniform matroids.
We now show that the remaining 6 extreme rays contain representable polymatroids and so are almost entropic. For the ray containing
$$\mathbf{h_1}(A)=\left\{
\begin{aligned}
	2 \quad& \text{if}\ |A| = 1,\\
	4 \quad& \text{if}\ |A| = 2 ,\\
	5 \quad& \text{if}\ A \in \mathcal{O}(124),\\
	6 \quad& \text{otherwise},
	%6 \quad& \text{if}\ A \in \mathcal{O}(123)\ \text{or}\ |A| \geq 4.
\end{aligned}
\right.
$$
as it is the sum of two Fano matroids\footnote{The definition and its representation of Fano matroid can be found in \textnormal{\cite[Chapter 6]{oxley2006matroid}}}, which are representable, $\mathbf{h_1}$ is representable as well. Now consider the ray containing
%is the sum of two Fano matroids\footnote{The representation of Fano matroid can be found in \textnormal{\cite[Chapter 6]{oxley2006matroid}}.}. Since Fano matroids are representable, $\mathbf{h_1}$ is also representable. Besides, one find the polymatroid
$$\mathbf{h_2}(A)=\left\{
\begin{aligned}
	2 \quad& \text{if}\ |A| = 1,\\
	4 \quad& \text{if}\ |A| = 2 ,\\
	6 \quad& \text{if}\ |A| = 3 ,\\
	7 \quad& \text{if}\ A \in \mathcal{O}(1235),\\
	8 \quad& \text{otherwise}.
	%8 \quad& \text{if}\ A \in \mathcal{O}(1234)\ \text{or}\ |A| \geq 5.
\end{aligned}
\right.
$$
Note that $\mathbf{h_2}=\mathbf{r_1}+\mathbf{r_2}$, where 
$$\mathbf{r_1}(A)=\left\{
\begin{aligned}
	1 \quad& \text{if}\ |A| = 1,\\
	2 \quad& \text{if}\ |A| = 2 ,\\
	3 \quad& \text{if}\ |A| = 3\ \text{or} \ A = \{1,2,3,5\},\{1,2,4,7\},\\ & \{1,3,6,7\},\{1,4,5,6\},\{2,3,4,6\},\{2,5,6,7\},\\ & \{3,4,5,7\},\\
	4 \quad& \text{otherwise},
	%8 \quad& \text{if}\ A \in \mathcal{O}(1234)\ \text{or}\ |A| \geq 5.
\end{aligned}
\right.
$$
$$\mathbf{r_2}(A)=\left\{
\begin{aligned}
	1 \quad& \text{if}\ |A| = 1,\\
	2 \quad& \text{if}\ |A| = 2 ,\\
	3 \quad& \text{if}\ |A| = 3\ \text{or} \ A = \{1,2,3,6\},\{1,2,5,7\},\\ & \{1,3,4,5\}, \{1,4,6,7\},\{2,3,4,7\},\{2,4,5,6\},\\ & \{3,5,6,7\},\\
	4 \quad& \text{otherwise},
	%8 \quad& \text{if}\ A \in \mathcal{O}(1234)\ \text{or}\ |A| \geq 5.
\end{aligned}
\right.
$$
%can be split into two polymatroids both with seven elements, 
%thus is representable by Lemma \ref{lem4}.
are both matroids with $7$ elements. By \textnormal{\cite[Prop. 6.4.10]{oxley2006matroid}},
$\mathbf{r_1}$ and $\mathbf{r_2}$ are representable and then so is $\mathbf{h_2}$. For the remaining $4$ extreme rays, due to the page limitation, we list the polymatroids they contain in Appendix A.
Using $\mathsf{ITAP}$\cite{ITAP}, we obtain the representations of the four polymatroids 
 over $\mathrm{GF}(2)$, as shown in the following matrices.
\[
\begin{bNiceArray}{@{}ccccccccc@{}}[first-row]
	&1  &2 &3 &4 &5 &6 &7 &\\
	&1\  0 &0\  0&0\  0&1\  0 &0\  0&1\  0 &1\  1 & \\
	&0\  1 &0\  0&1\  0&0\  0 &0\  1&1\  1 &0\  0 & \\
	&0\  0 &1\  0&1\  0&1\  0 &0\  0&0\  0 &1\  1 & \\
	&0\  0 &0\  1&0\  0&0\  0 &1\  0&0\  1 &1\  0 & \\
	&0\  0 &0\  0&0\  0&0\  1 &1\  0&1\  0 &0\  1 & \\
	&0\  0 &0\  0&0\  1&0\  0 &0\  1&0\  1 &0\  1 & \\
	\CodeAfter
	\tikz \draw [dashed, line width=0.6pt, black] (1-|3) -- (7-|3) ;
	\tikz \draw [dashed, line width=0.6pt, black] (1-|4) -- (7-|4) ;
	\tikz \draw [dashed, line width=0.6pt, black] (1-|5) -- (7-|5) ;
	\tikz \draw [dashed, line width=0.6pt, black] (1-|6) -- (7-|6) ;
	\tikz \draw [dashed, line width=0.6pt, black] (1-|7) -- (7-|7) ;
	\tikz \draw [dashed, line width=0.6pt, black] (1-|8) -- (7-|8) ;
\end{bNiceArray}
\] 
\[
\begin{bNiceArray}{@{}ccccccccc@{}}[first-row]
	&1  &2 &3 &4 &5 &6 &7 &\\
	&1\  0 &0\  0&1\  0&1\  0 &0\  0&1\  1 &1\  1 & \\
	&0\  1 &0\  0&0\  0&0\  0 &0\  0&0\  1 &0\  0 & \\
	&0\  0 &1\  0&1\  0&1\  0 &0\  0&0\  0 &1\  1 & \\
	&0\  0 &0\  1&1\  1&0\  0 &0\  0&0\  1 &1\  0 & \\
	&0\  0 &0\  0&0\  0&1\  0 &1\  0&0\  0 &0\  1 & \\
	&0\  0 &0\  0&0\  1&0\  0 &0\  1&1\  0 &0\  1 & \\
	&0\  0 &0\  0&1\  0&0\  1 &0\  0&1\  0 &0\  1 & \\
	&0\  0 &0\  0&0\  1&0\  0 &0\  0&0\  1 &0\  1 & \\
	\CodeAfter
	\tikz \draw [dashed, line width=0.6pt, black] (1-|3) -- (10-|3) ;
	\tikz \draw [dashed, line width=0.6pt, black] (1-|4) -- (10-|4) ;
	\tikz \draw [dashed, line width=0.6pt, black] (1-|5) -- (10-|5) ;
	\tikz \draw [dashed, line width=0.6pt, black] (1-|6) -- (10-|6) ;
	\tikz \draw [dashed, line width=0.6pt, black] (1-|7) -- (10-|7) ;
	\tikz \draw [dashed, line width=0.6pt, black] (1-|8) -- (10-|8) ;
\end{bNiceArray}
\] 
\[
\begin{bNiceArray}{@{}ccccccccc@{}}[first-row]
	&1       &2       &3      &4      &5       &6      &7& \\
	&1\ 1\ 0 &1\ 0\ 0 &1\ 0\ 0&0\ 0\ 0&1\ 0\ 0&1\ 0\ 0&0\ 0\ 0 & \\
	&0\ 0\ 0 &0\ 0\ 0 &1\ 0\ 0&0\ 0\ 0&0\ 1\ 0&0\ 1\ 0&0\ 0\ 0 & \\
	&0\ 0\ 0 &0\ 0\ 0 &0\ 0\ 0&1\ 0\ 0&0\ 0\ 1&0\ 0\ 0&1\ 0\ 0 & \\
	&0\ 1\ 0 &0\ 0\ 0 &0\ 1\ 0&0\ 0\ 0&0\ 0\ 0&0\ 0\ 0&0\ 1\ 0 & \\
	&1\ 0\ 0 &0\ 0\ 0 &0\ 0\ 0&0\ 0\ 0&0\ 1\ 1&0\ 0\ 1&0\ 0\ 1 & \\
	&0\ 0\ 0 &0\ 1\ 0 &0\ 0\ 0&0\ 1\ 0&0\ 1\ 0&1\ 0\ 1&0\ 1\ 0 & \\
	&0\ 0\ 1 &0\ 1\ 0 &0\ 0\ 0&0\ 1\ 0&1\ 0\ 0&1\ 0\ 0&0\ 1\ 0 & \\
	&0\ 0\ 0 &0\ 1\ 0 &0\ 0\ 1&0\ 0\ 0&0\ 0\ 0&0\ 0\ 0&1\ 0\ 1 & \\
	&0\ 0\ 0 &0\ 0\ 0 &0\ 0\ 0&0\ 0\ 1&0\ 0\ 1&0\ 1\ 1&0\ 1\ 1 & \\
	&0\ 0\ 0 &0\ 0\ 1 &0\ 0\ 0&0\ 0\ 0&0\ 0\ 0&1\ 0\ 1&1\ 1\ 0 & \\
	\CodeAfter
	\tikz \draw [dashed, line width=0.6pt, black] (1-|3) -- (11-|3) ;
	\tikz \draw [dashed, line width=0.6pt, black] (1-|4) -- (11-|4) ;
	\tikz \draw [dashed, line width=0.6pt, black] (1-|5) -- (11-|5) ;
	\tikz \draw [dashed, line width=0.6pt, black] (1-|6) -- (11-|6) ;
	\tikz \draw [dashed, line width=0.6pt, black] (1-|7) -- (11-|7) ;
	\tikz \draw [dashed, line width=0.6pt, black] (1-|8) -- (11-|8) ;
\end{bNiceArray}
\]
\[
\begin{bNiceArray}{@{}ccccccccc@{}}[first-row]
	&1       &2       &3      &4  &5  &6      &7& \\
	&1\ 0\ 0 &1\ 0\ 0 &0\ 0\ 0&1\ 1\ 0&0\ 0\ 0&0\ 0\ 0&0\ 1\ 1 & \\
	&0\ 1\ 0 &0\ 0\ 0 &0\ 0\ 0&1\ 1\ 0&1\ 0\ 0&1\ 0\ 0&1\ 0\ 0 & \\
	&0\ 0\ 0 &1\ 0\ 0 &0\ 0\ 0&1\ 0\ 0&0\ 0\ 0&0\ 1\ 0&0\ 1\ 0 & \\
	&0\ 0\ 0 &0\ 0\ 0 &0\ 0\ 0&1\ 0\ 0&1\ 0\ 0&0\ 0\ 1&0\ 0\ 1 & \\
	&0\ 0\ 0 &0\ 1\ 0 &1\ 0\ 0&1\ 0\ 0&0\ 0\ 0&0\ 0\ 0&0\ 1\ 0 & \\
	&0\ 0\ 0 &0\ 0\ 0 &1\ 0\ 0&1\ 0\ 0&0\ 1\ 0&0\ 0\ 0&0\ 1\ 1 & \\
	&0\ 0\ 0 &0\ 0\ 0 &0\ 1\ 0&0\ 1\ 0&0\ 0\ 0&0\ 0\ 1&0\ 1\ 0 & \\
	&0\ 0\ 1 &0\ 0\ 0 &0\ 0\ 0&0\ 0\ 1&0\ 0\ 0&0\ 0\ 0&1\ 0\ 1 & \\
	&0\ 0\ 0 &0\ 0\ 1 &0\ 0\ 0&0\ 0\ 1&0\ 0\ 0&0\ 0\ 0&0\ 1\ 0 & \\
	&0\ 0\ 0 &0\ 0\ 0 &0\ 0\ 1&0\ 0\ 0&0\ 0\ 0&0\ 0\ 1&1\ 0\ 0 & \\
	&0\ 0\ 0 &0\ 0\ 0 &0\ 0\ 0&0\ 0\ 0&0\ 0\ 1&1\ 0\ 0&0\ 0\ 1 & \\
	\CodeAfter
	\tikz \draw [dashed, line width=0.6pt, black] (1-|3) -- (12-|3) ;
	\tikz \draw [dashed, line width=0.6pt, black] (1-|4) -- (12-|4) ;
	\tikz \draw [dashed, line width=0.6pt, black] (1-|5) -- (12-|5) ;
	\tikz \draw [dashed, line width=0.6pt, black] (1-|6) -- (12-|6) ;
	\tikz \draw [dashed, line width=0.6pt, black] (1-|7) -- (12-|7) ;
	\tikz \draw [dashed, line width=0.6pt, black] (1-|8) -- (12-|8) ;
\end{bNiceArray}
\]
\begin{spacing}{1.2}
\noindent\underline{$\overline{\Psi_G^*} \neq \Psi_G$ for $G=S_2\times S_5$, $S_3\times S_4$:} 
\end{spacing}
It is immediately implied  by Theorem \ref{lem1}.
	
\noindent\underline{$\overline{\Psi_G^*} \neq \Psi_G$ for $G=S_1 \times \mathrm{PSL}_2(5)$:} 

Among the extreme rays of $\Psi_G$, one of them containing the following polymatroid 
$$ \mathbf{h_3}(A)=\left\{
\begin{aligned}
2 \quad& \text{if}\ |A| = 1 \ \text{and} \ A \neq \{1\},\\
3 \quad& \text{if}\ A = \{1\}, \\
4 \quad& \text{if}\ |A| = 2 ,\\
5 \quad& \text{if}\ A \in \mathcal{O}(123)\ \text{or}\ \mathcal{O}(234),\\
6 \quad& \text{otherwise}.
%6 \quad& \text{if}\ A \in \mathcal{O}(235)\ \text{or}\ \lvert A \rvert \geq 4.
\end{aligned}
\right.
$$
By contracting $\{2\}$ from $\mathbf{h_3}$, then restricting it on $\{1,3,4,5\}$,
we obtain $\widetilde{\mathbf{h_3}}$, which violates Zhang-Yeung inequality. Thus, it is non-entropic, $\overline{\Psi_G^*} \neq \Psi_G$.
%$$ \widetilde{\mathbf{h_3}}(A)=\left\{
%\begin{aligned}
%2 \quad& \text{if}\ |A| = 1, \\
%3 \quad& \text{if}\ |A| = 2 \ \text{and} \ A \neq \{3,5\},\\
%4 \quad& \text{if}\ A = \{3,5\}\ \text{or}\ |A| \geq 3.
%\end{aligned}
%\right.
%$$
\begin{spacing}{1.2}
\noindent\underline{$\overline{\Psi_G^*} \neq \Psi_G$ for $G = S_1 \times  S_3\mathrm{wr}_2C_2$, $S_1 \times  S_2\mathrm{wr}_3S_3$:} 
\end{spacing}
Since $\overline{\Psi_{S_3\mathrm{wr}_2C_2}^*} \neq \Psi_{S_3\mathrm{wr}_2C_2}$, $\overline{\Psi_{S_2\mathrm{wr}_3S_3}^*} \neq \Psi_{S_2\mathrm{wr}_3S_3}$, which immediately implies the result by Lemma \ref{Lem2}.

%Among all the extreme rays of $\Psi_{S_1 \times  S_3\mathrm{wr}_2C_2}$, one find the polymatroid
%$$ \mathbf{h_2}(A)=\left\{
%\begin{aligned}
%2 \quad& \text{if}\ \lvert A \rvert = 1 ,\\
%4 \quad& \text{if}\ |A| = 2\  \text{or}\ A \in \mathcal{O}(257),\\
%5 \quad& \text{if}\ A \in \mathcal{O}(123)\ \text{or}\ \mathcal{O}(234)\ \text{or}\ \mathcal{O}(2346),\\
%6 \quad& \text{otherwise}.
%6 \quad& \text{if}\ A \in \mathcal{O}(125)\ \text{or}\ |A| \geq 4\  \text{and}\  A \notin \mathcal{O}(2346).
%\end{aligned}
%\right.
%$$
%By contracting $\{2\}$ from $\mathbf{h_2}$, then restricting it on $\{1,3,4,5\}$, we obtain $\widetilde{\mathbf{h_2}}$, which violates Zhang-Yeung inequality, thus non-entropic.
%$$\widetilde{\mathbf{h_2}}(A)=\left\{
%%\begin{aligned}
%%2 \quad& \text{if}\ \lvert A \rvert = 1, \\
%%3 \quad& \text{if}\ \lvert A \rvert = 2 \ \text{and} \ |A| \neq \{1,4\},\\
%%4 \quad& \text{if}\ \lvert A \rvert = \{1,4\}\ \text{or}\ \lvert A \rvert \geq 3.
%%\end{aligned}
%\right.
%$$
\begin{spacing}{1.2}
\noindent\underline{$\overline{\Psi_G^*} \neq \Psi_G$ for $G=\mathrm{PGL}_3(2)$:}  
\end{spacing}
Among the extreme rays of $\Psi_G$, there exsits one containing the following polymatroid
$$\mathbf{h_4}(A)=\left\{
\begin{aligned}
2 \quad& \text{if}\ |A| = 1,\\
4 \quad& \text{if}\ |A| = 2 ,\\
5 \quad& \text{if}\ A \in \mathcal{O}(123),\\
6 \quad& \text{otherwise}.
%6 \quad& \text{if}\ A \in \mathcal{O}(126)\ \text{or}\ |A| \geq 4.
\end{aligned}
\right.
$$
By contracting $\{1\}$ from $\mathbf{h_4}$, then restricting it on $\{2,3,4,5\}$, we obtain $\widetilde{\mathbf{h_4}}$, which violates Zhang-Yeung inequality. Thus, it is non-entropic and $\overline{\Psi_G^*} \neq \Psi_G$.
\end{proof}

\bibliography{isit2024}
\bibliographystyle{IEEEtran}

\clearpage

%\onecolumn
\section*{Appendix}

\setcounter{subsection}{0}
\subsection{More details about the symmetric polymatroidal regions}
\label{appa}
%	\begin{figure*}[b]
%	\centering
%	\vspace{-1cm}
%	\twocolumn[{\includegraphics[width=0.9\linewidth,height=0.5\linewidth]{6order.pdf}}]
%	\captionsetup{font=small}
%	\caption{Summary of results for $n=6$}
%	\label{fig3}
%\end{figure*}
In this subsection of Appendix, we provide more details about the symmetric polymatroidal regions, which are used in the proofs in Section \ref{results}, i.e., orbit structures,  H-representation and V-representation of $\Psi_G$.
\subsection*{1). $G=\mathrm{PSL}_2(5):\langle (123)(456),(14623) \rangle$}
We depicted the orbit structure of $\mathrm{PSL}_2(5)$ in Fig. \ref{fig6}.
\setcounter{figure}{0}
\begin{figure}[htp]
	\centering
	\includegraphics[scale=0.52]{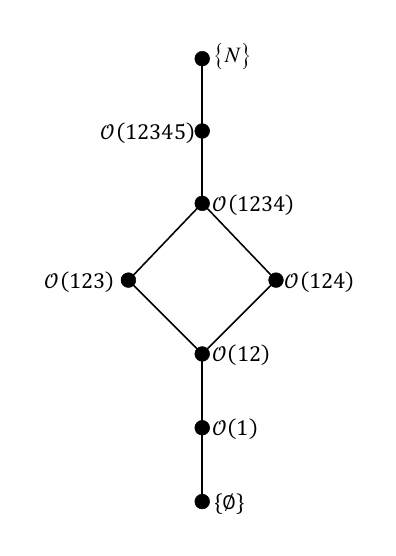}
	\captionsetup{font=small}
	\caption{$\textfrak{O}_{\mathrm{PSL}_2(5)}$}
	\label{fig6}
\end{figure}

Note that each orbit in $\textfrak{O}_{\mathrm{PSL}_2(5)}$ corresponds a dimension in $\mathrm{fix}_{\mathrm{PSL}_2(5)}$, so
\begin{align}
	\hspace{-2cm}
	\Psi_{\mathrm{PSL}_2(5)} =\{& \mathbf{s} \in \mathrm{fix}_{\mathrm{PSL}_2(5)}: \notag 2\mathbf{s}_{12}\geq\mathbf{s}_1 +\mathbf{s}_{123},\\ \notag
	&2\mathbf{s}_{12}\geq\mathbf{s}_1 +\mathbf{s}_{124},\\ \notag
	&2\mathbf{s}_{124}\geq\mathbf{s}_{12} +\mathbf{s}_{1234},\\ \notag
	&\mathbf{s}_{123}+\mathbf{s}_{124}\geq\mathbf{s}_{12} +\mathbf{s}_{1234},\\ \notag
	&2\mathbf{s}_{123}\geq\mathbf{s}_{12} +\mathbf{s}_{1234},\\ \notag
	&2\mathbf{s}_{1234}\geq\mathbf{s}_{123} +\mathbf{s}_{12345},\\ \notag
	&2\mathbf{s}_{1234}\geq\mathbf{s}_{124} +\mathbf{s}_{12345},\\ \notag
	&2\mathbf{s}_{12345}\geq\mathbf{s}_{1234} +\mathbf{s}_{123456},\\ \notag
	&\mathbf{s}_{123456}\geq\mathbf{s}_{12345},\\ \notag
	&2\mathbf{s}_{1}\geq\mathbf{s}_{12} \}.
\end{align}
For simplicity, we rewrite each inequality by its coefficient, and so the H-representation of $\Psi_{\mathrm{PSL}_2(5)}$ is
\begin{itemize}
	\item $-1$  $2$  $-1$  $0$  $0$  $0$  $0$
	\item$-1$  $2$  $0$  $-1$  $0$  $0$  $0$
	\item$0$  $-1$  $0$  $2$  $-1$  $0$  $0$
	\item$0$  $-1$  $1$  $1$  $-1$  $0$  $0$
	\item$0$  $-1$  $2$  $0$  $-1$  $0$  $0$
	\item$0$  $0$  $-1$  $0$  $2$  $-1$  $0$
	\item$0$  $0$  $0$  $-1$  $2$  $-1$  $0$
	\item$0$  $0$  $0$  $0$  $-1$  $2$  $-1$
	\item$0$  $0$  $0$  $0$  $0$  $-1$  $1$
	\item$2$  $-1$  $0$  $0$  $0$  $0$  $0$
\end{itemize}
where the order of the indices is $(\mathcal{O}(1)$,$\mathcal{O}(12)$,$\mathcal{O}(123)$, $\mathcal{O}(124)$,$\mathcal{O}(1234)$, $\mathcal{O}(12345)$,$\mathcal{O}(123456))$.

For the V-representation, for those extreme rays containing polymatroid other than uniform matroid, we represent it by a vector indexed by the orbits.
\begin{itemize}
	\item $U_{1,6}$, $U_{2,6}$, $U_{3,6}$, $U_{4,6}$, $U_{5,6}$, $U_{6,6}$;
	\item $2$  $4$  $5$  $6$  $6$  $6$  $6$;
	\item $2$  $4$  $6$  $5$  $6$  $6$  $6$;
\end{itemize}
%All the $19$ extreme rays of $\Psi_{S_3\mathrm{wr}_2C_2}$:

In the following, for a $G$-symmetric polymatroidal region $\Psi_G$, we represent its orbit structure by its Hasse diagram in the figure, the H-representation by inequality coefficients and V-representation by vector of polymatroid except the uniform matroid.
\subsection*{2). $G=S_3\mathrm{wr}_2C_2:\langle (12)(34)(56),(14)(235) \rangle$ }
\begin{figure}[htp]
	\centering
	\includegraphics[scale=0.52]{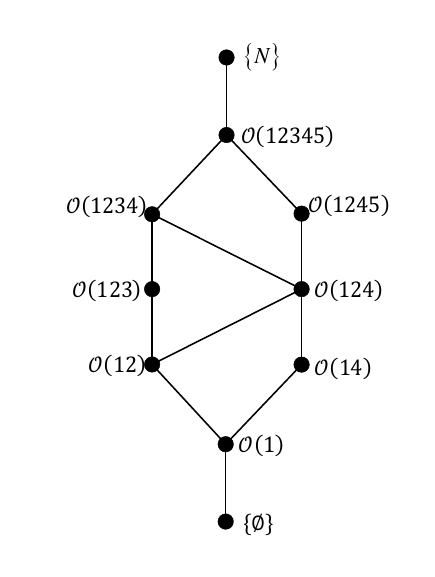}
	\captionsetup{font=small}
	\caption{$\textfrak{O}_{S_3\mathrm{wr}_2C_2}$}
	\label{fig7}
\end{figure}
The order of indices: $(\mathcal{O}(1)$,$\mathcal{O}(12)$,$\mathcal{O}(14)$,$\mathcal{O}(124)$,$\mathcal{O}(146)$, $\mathcal{O}(1234)$,$\mathcal{O}(1246)$,$\mathcal{O}(12345)$,$\mathcal{O}(123456))$.

H-representation:
\begin{itemize}
	\item $-1$  $0$  $2$  $-1$  $0$  $0$  $0$  $0$  $0$
	\item $-1$  $1$  $1$  $-1$  $0$  $0$  $0$  $0$  $0$
	\item $-1$  $2$  $0$  $0$  $-1$  $0$  $0$  $0$  $0$
	\item $0$  $-1$  $0$  $1$  $1$  $-1$  $0$  $0$  $0$
	\item $0$  $-1$  $0$  $2$  $0$  $0$  $-1$  $-1$  $0$
	\item $0$  $-1$  $0$  $2$  $0$  $0$  $-1$  $0$  $0$
	\item $0$  $0$  $-1$  $2$  $0$  $-1$  $0$  $0$  $0$
	\item $0$  $0$  $-1$  $2$  $0$  $0$  $-1$  $-1$  $0$
	\item $0$  $0$  $-1$  $2$  $0$  $0$  $-1$  $0$  $0$
	\item $0$  $0$  $0$  $-1$  $0$  $0$  $2$  $-1$  $0$
	\item $0$  $0$  $0$  $-1$  $0$  $0$  $2$  $0$  $0$
	\item $0$  $0$  $0$  $-1$  $0$  $1$  $1$  $-1$  $0$
	\item $0$  $0$  $0$  $-1$  $0$  $1$  $1$  $0$  $0$
	\item $0$  $0$  $0$  $0$  $-1$  $2$  $0$  $-1$  $0$
	\item $0$  $0$  $0$  $0$  $0$  $-1$  $0$  $2$  $-1$
	\item $0$  $0$  $0$  $0$  $0$  $0$  $-1$  $1$  $-1$
	\item $0$  $0$  $0$  $0$  $0$  $0$  $-1$  $2$  $-1$
	\item $0$  $0$  $0$  $0$  $0$  $0$  $0$  $-1$  $1$
	\item $2$  $-1$  $0$  $0$  $0$  $0$  $0$  $0$  $0$
	\item $2$  $0$  $-1$  $0$  $0$  $0$  $0$  $0$  $0$
\end{itemize}

V-representation:
\begin{itemize}
	\item $U_{1,6}$, $U_{2,6}$, $U_{3,6}$, $U_{4,6}$, $U_{5,6}$,$U_{6,6}$; 
	\item $2$  $3$  $4$  $4$  $4$  $4$  $4$  $4$  $4$
	\item $2$  $4$  $3$  $4$  $4$  $4$  $4$  $4$  $4$
	\item $2$  $4$  $4$  $5$  $6$  $6$  $6$  $6$  $6$
	\item $4$  $7$  $8$  $10$  $12$  $12$  $12$  $12$  $12$
	\item $1$  $2$  $2$  $3$  $2$  $3$  $3$  $3$  $3$
	\item $1$  $2$  $1$  $2$  $1$  $2$  $2$  $2$  $2$
	\item $3$  $5$  $6$  $7$  $9$  $8$  $9$  $9$  $9$
	\item $4$  $8$  $8$  $10$  $12$  $11$  $12$  $12$  $12$
	\item $3$  $5$  $6$  $7$  $8$  $8$  $9$  $9$  $9$
	\item $2$  $4$  $4$  $6$  $6$  $7$  $8$  $8$  $8$
	\item $1$  $2$  $2$  $3$  $2$  $4$  $3$  $4$  $4$
	\item $2$  $4$  $3$  $5$  $4$  $6$  $5$  $6$  $6$
	\item $2$  $4$  $4$  $6$  $6$  $8$  $7$  $8$  $8$
\end{itemize}
%All the $18$ extreme rays of $\Psi_{S_2\mathrm{wr}_3S_3}$:
\subsection*{3). $S_2\mathrm{wr}_3S_3:\langle (123456),(16)(34) \rangle$ }
\begin{figure}[htp]
	\centering
	\includegraphics[scale=0.52]{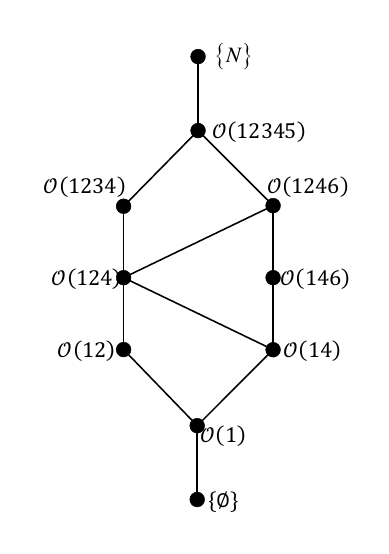}
	\captionsetup{font=small}
	\caption{$\textfrak{O}_{S_2\mathrm{wr}_3S_3}$}
	\label{fig8}
\end{figure}
The order of indices: $(\mathcal{O}(1)$,$\mathcal{O}(12)$,$\mathcal{O}(14)$,$\mathcal{O}(123)$,$\mathcal{O}(124)$, $\mathcal{O}(1234)$,$\mathcal{O}(1245)$,$\mathcal{O}(12345)$,$\mathcal{O}(123456))$.

H-representation:
\begin{itemize}
	\item $-1$  $1$  $1$  $0$  $-1$  $0$  $0$  $0$  $0$
	\item $-1$  $2$  $0$  $-1$  $0$  $0$  $0$  $0$  $0$
	\item $-1$  $2$  $0$  $0$  $-1$  $0$  $0$  $0$  $0$
	\item $0$  $-1$  $0$  $0$  $2$  $0$  $-1$  $0$  $0$
	\item $0$  $-1$  $0$  $1$  $1$  $-1$  $0$  $-1$  $0$
	\item $0$  $-1$  $0$  $1$  $1$  $-1$  $0$  $0$  $0$
	\item $0$  $-1$  $0$  $2$  $0$  $-1$  $0$  $-1$  $0$
	\item $0$  $-1$  $0$  $2$  $0$  $-1$  $0$  $0$  $0$
	\item $0$  $0$  $-1$  $0$  $2$  $-1$  $0$  $-1$  $0$
	\item $0$  $0$  $-1$  $0$  $2$  $-1$  $0$  $0$  $0$
	\item $0$  $0$  $-1$  $0$  $2$  $0$  $-1$  $0$  $0$
	\item $0$  $0$  $0$  $-1$  $0$  $2$  $0$  $-1$  $0$
	\item $0$  $0$  $0$  $-1$  $0$  $2$  $0$  $0$  $0$
	\item $0$  $0$  $0$  $0$  $-1$  $1$  $1$  $-1$  $0$
	\item $0$  $0$  $0$  $0$  $-1$  $1$  $1$  $0$  $0$
	\item $0$  $0$  $0$  $0$  $-1$  $2$  $0$  $-1$  $0$
	\item $0$  $0$  $0$  $0$  $-1$  $2$  $0$  $0$  $0$
	\item $0$  $0$  $0$  $0$  $0$  $-1$  $0$  $1$  $-1$
	\item $0$  $0$  $0$  $0$  $0$  $-1$  $0$  $2$  $-1$
	\item $0$  $0$  $0$  $0$  $0$  $0$  $-1$  $2$  $-1$
	\item $0$  $0$  $0$  $0$  $0$  $0$  $0$  $-1$  $1$
	\item $2$  $-1$  $0$  $0$  $0$  $0$  $0$  $0$  $0$
	\item $2$  $0$  $-1$  $0$  $0$  $0$  $0$  $0$  $0$
\end{itemize}

V-representation:
\begin{itemize}
	\item $U_{1,6}$,$U_{2,6}$, $U_{3,6}$, $U_{4,6}$, $U_{5,6}$, $U_{6,6}$; 
	\item $2$  $3$  $4$  $4$  $4$  $4$  $4$  $4$  $4$
	\item $1$  $2$  $1$  $2$  $2$  $2$  $2$  $2$  $2$
	\item $4$  $6$  $8$  $7$  $8$  $8$  $8$  $8$  $8$
	\item $2$  $4$  $4$  $5$  $6$  $6$  $6$  $6$  $6$
	\item $2$  $4$  $4$  $6$  $5$  $6$  $6$  $6$  $6$
	\item $2$  $4$  $3$  $6$  $5$  $6$  $6$  $6$  $6$
	\item $2$  $4$  $4$  $6$  $6$  $7$  $8$  $8$  $8$
	\item $4$  $8$  $8$  $11$  $12$  $14$  $16$  $16$  $16$
	\item $2$  $4$  $4$  $6$  $5$  $6$  $5$  $6$  $6$
	\item $1$  $2$  $1$  $3$  $2$  $3$  $2$  $3$  $3$
	\item $2$  $4$  $2$  $5$  $4$  $6$  $4$  $6$  $6$
	\item $1$  $2$  $2$  $3$  $3$  $4$  $3$  $4$  $4$
\end{itemize}

%All the $13$ extreme rays of $\Psi_{\mathrm{AGL}_1(7)}$:
\subsection*{4). $\mathrm{AGL}_1(7):\langle (1234567),(163247) \rangle$  }
\begin{figure}[htp]
	\centering
	\includegraphics[scale=0.52]{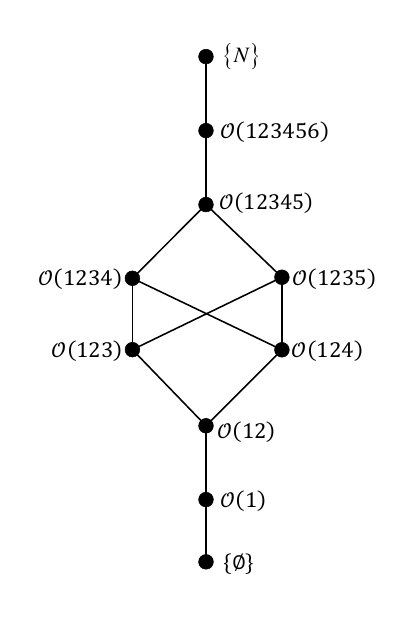}
	\captionsetup{font=small}
	\caption{$\textfrak{O}_{\mathrm{AGL}_1(7)}$}
	\label{fig9}
\end{figure}
The order of indices: $(\mathcal{O}(1)$,$\mathcal{O}(12)$,$\mathcal{O}(123)$,$\mathcal{O}(124)$,$\mathcal{O}(1234)$, $\mathcal{O}(1235)$,$\mathcal{O}(12345)$,$\mathcal{O}(123456)$,$\mathcal{O}(1234567))$.

H-representation:
\begin{itemize}
	\item $-1$  $2$  $-1$  $0$  $0$  $0$  $0$  $0$  $0$
	\item$-1$  $2$  $0$  $-1$  $0$  $0$  $0$  $0$  $0$
	\item$0$  $-1$  $0$  $2$  $-1$  $0$  $0$  $0$  $0$
	\item$0$  $-1$  $1$  $1$  $-1$  $0$  $0$  $0$  $0$
	\item$0$  $-1$  $1$  $1$  $0$  $-1$  $0$  $0$  $0$
	\item$0$  $-1$  $2$  $0$  $-1$  $0$  $0$  $0$  $0$
	\item$0$  $-1$  $2$  $0$  $0$  $-1$  $0$  $0$  $0$
	\item$0$  $0$  $-1$  $0$  $0$  $2$  $-1$  $0$  $0$
	\item$0$  $0$  $-1$  $0$  $1$  $1$  $-1$  $0$  $0$
	\item$0$  $0$  $-1$  $0$  $2$  $0$  $-1$  $0$  $0$
	\item$0$  $0$  $0$  $-1$  $1$  $1$  $-1$  $0$  $0$
	\item$0$  $0$  $0$  $-1$  $2$  $0$  $-1$  $0$  $0$
	\item$0$  $0$  $0$  $0$  $-1$  $0$  $2$  $-1$  $0$
	\item$0$  $0$  $0$  $0$  $0$  $-1$  $2$  $-1$  $0$
	\item$0$  $0$  $0$  $0$  $0$  $0$  $-1$  $2$  $-1$
	\item$0$  $0$  $0$  $0$  $0$  $0$  $0$  $-1$  $1$
	\item$2$  $-1$  $0$  $0$  $0$  $0$  $0$  $0$  $0$
\end{itemize}

V-representation:
\begin{itemize}
	\item $U_{1,7}$, $U_{2,7}$, $U_{3,7}$, $U_{4,7}$,$U_{5,7}$, $U_{6,7}$, $U_{7,7}$;
	\item $2$  $4$  $4$  $5$  $6$  $6$  $6$  $6$  $6$
	\item $2$  $4$  $4$  $6$  $8$  $7$  $8$  $8$  $8$
	\item $2$  $4$  $4$  $6$  $6$  $6$  $6$  $6$  $6$
	\item $3$  $6$  $6$  $9$  $10$  $9$  $10$  $10$  $10$
	\item $2$  $4$  $4$  $6$  $7$  $8$  $8$  $8$  $8$
	\item $3$  $6$  $6$  $8$  $10$  $11$  $11$  $11$  $11$
\end{itemize}

%All the $13$ extreme rays of $\Psi_{\mathrm{PGL}_3(2)}$:
\subsection*{ 5). $\Psi_{\mathrm{PGL}_3(2)}:\langle (1234567),(1367)(45) \rangle$ }
\begin{figure}[htp]
	\centering
	\includegraphics[scale=0.52]{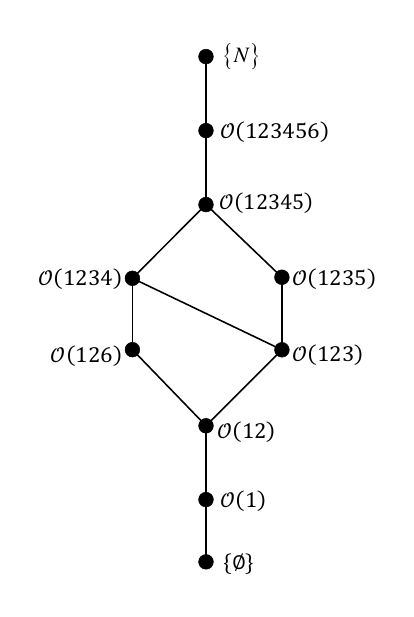}
	\captionsetup{font=small}
	\caption{$\textfrak{O}_{\mathrm{PGL}_3(2)}$}
	\label{fig10}
\end{figure}
The order of indices: $(\mathcal{O}(1)$,$\mathcal{O}(12)$,$\mathcal{O}(123)$,$\mathcal{O}(126)$,$\mathcal{O}(1234)$, $\mathcal{O}(1235)$,$\mathcal{O}(12345)$,$\mathcal{O}(123456)$,$\mathcal{O}(1234567))$.

H-representation:
\begin{itemize}
	\item $ -1$  $2$  $-1$  $0$  $0$  $0$  $0$  $0$  $0$
	\item $-1$  $2$  $0$  $-1$  $0$  $0$  $0$  $0$  $0$
	\item $0$  $-1$  $1$  $1$  $-1$  $0$  $0$  $0$  $0$
	\item $0$  $-1$  $2$  $0$  $-1$  $0$  $0$  $0$  $0$
	\item $0$  $-1$  $2$  $0$  $0$  $-1$  $0$  $0$  $0$
	\item $0$  $0$  $-1$  $0$  $1$  $1$  $-1$  $0$  $0$
	\item $0$  $0$  $-1$  $0$  $2$  $0$  $-1$  $0$  $0$
	\item $0$  $0$  $0$  $-1$  $2$  $0$  $-1$  $0$  $0$
	\item $0$  $0$  $0$  $0$  $-1$  $0$  $2$  $-1$  $0$
	\item $0$  $0$  $0$  $0$  $0$  $-1$  $2$  $-1$  $0$
	\item $0$  $0$  $0$  $0$  $0$  $0$  $-1$  $2$  $-1$
	\item $0$  $0$  $0$  $0$  $0$  $0$  $0$  $-1$  $1$
	\item $2$  $-1$  $0$  $0$  $0$  $0$  $0$  $0$  $0$
\end{itemize}

V-representation:
\begin{itemize}
	\item $U_{1,7}$, $U_{2,7}$, $U_{3,7}$, $U_{4,7}$,$U_{5,7}$, $U_{6,7}$, $U_{7,7}$;
	\item $2$  $4$  $5$  $6$  $6$  $5$  $6$  $6$  $6$
	\item $1$  $2$  $3$  $2$  $3$  $3$  $3$  $3$  $3$
	\item $1$  $2$  $3$  $3$  $4$  $3$  $4$  $4$  $4$
	\item $2$  $4$  $5$  $6$  $6$  $6$  $6$  $6$  $6$
	\item $2$  $4$  $6$  $5$  $7$  $8$  $8$  $8$  $8$
	\item $2$  $4$  $6$  $6$  $7$  $8$  $8$  $8$  $8$
\end{itemize}

\subsection*{ 6). $ \Psi_{S_1\times\mathrm{PSL}_2(5)}:\langle (234)(567),(25734) \rangle$ }
\begin{figure}[htp]
	\centering
	\includegraphics[scale=0.52]{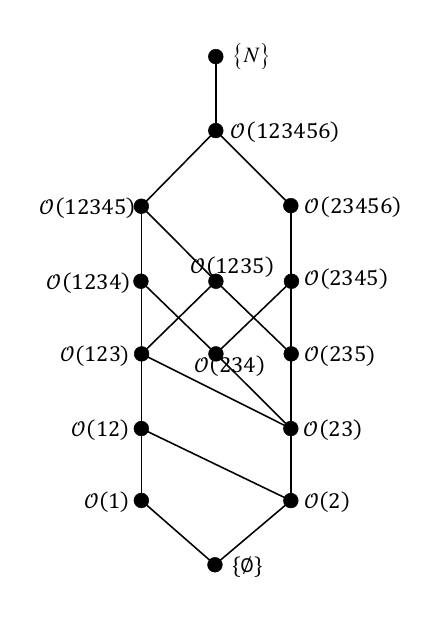}
	\captionsetup{font=small}
	\caption{$\textfrak{O}_{S_1\times\mathrm{PSL}_2(5)}$}
	\label{fig11}
\end{figure}
The order of indices: $(\mathcal{O}(1)$, $\mathcal{O}(2)$, $\mathcal{O}(12)$, $\mathcal{O}(23)$, $\mathcal{O}(123)$,
$\mathcal{O}(234)$, $\mathcal{O}(235)$, $\mathcal{O}(1234)$, $\mathcal{O}(1235)$, $\mathcal{O}(2345)$, $\mathcal{O}(12345)$,
$\mathcal{O}(23456)$,$\mathcal{O}(123456)$,$\mathcal{O}(1234567))$.

\begin{figure*}[b]
	\centering
	\includegraphics[width=0.95 \linewidth,height=0.6\linewidth]{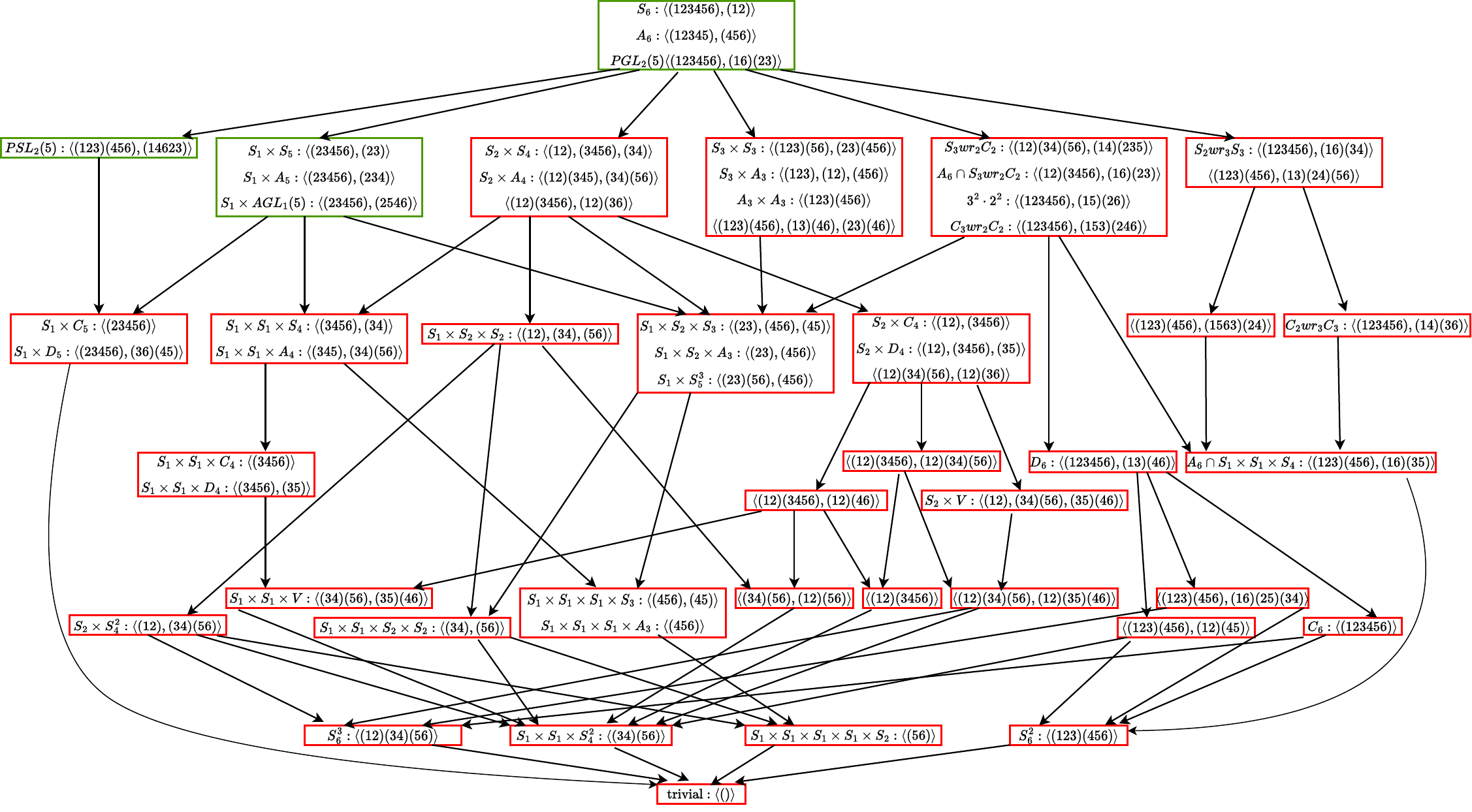}
	\captionsetup{font=small}
	\caption{The Hasse diagram of $P_6$}
	\label{fig4}
\end{figure*}
H-representation:
\begin{itemize}
	\item $-1$  $0$  $2$  $0$  $-1$  $0$  $0$  $0$  $0$  $0$  $0$  $0$  $0$  $0$
	\item $0$  $-1$  $0$  $2$  $0$  $-1$  $0$  $0$  $0$  $0$  $0$  $0$  $0$  $0$
	\item $0$  $-1$  $0$  $2$  $0$  $0$  $-1$  $0$  $0$  $0$  $0$  $0$  $0$  $0$
	\item $0$  $-1$  $1$  $1$  $-1$  $0$  $0$  $0$  $0$  $0$  $0$  $0$  $0$  $0$
	\item $0$  $0$  $-1$  $0$  $2$  $0$  $0$  $-1$  $0$  $0$  $0$  $0$  $0$  $0$
	\item $0$  $0$  $-1$  $0$  $2$  $0$  $0$  $0$  $-1$  $0$  $0$  $0$  $0$  $0$
	\item $0$  $0$  $0$  $-1$  $0$  $0$  $2$  $0$  $0$  $-1$  $0$  $0$  $0$  $0$
	\item $0$  $0$  $0$  $-1$  $0$  $1$  $1$  $0$  $0$  $-1$  $0$  $0$  $0$  $0$
	\item $0$  $0$  $0$  $-1$  $0$  $2$  $0$  $0$  $0$  $-1$  $0$  $0$  $0$  $0$
	\item $0$  $0$  $0$  $-1$  $1$  $0$  $1$  $0$  $-1$  $0$  $0$  $0$  $0$  $0$
	\item $0$  $0$  $0$  $-1$  $1$  $1$  $0$  $-1$  $0$  $0$  $0$  $0$  $0$  $0$
	\item $0$  $0$  $0$  $0$  $-1$  $0$  $0$  $0$  $2$  $0$  $-1$  $0$  $0$  $0$
	\item $0$  $0$  $0$  $0$  $-1$  $0$  $0$  $1$  $1$  $0$  $-1$  $0$  $0$  $0$
	\item $0$  $0$  $0$  $0$  $-1$  $0$  $0$  $2$  $0$  $0$  $-1$  $0$  $0$  $0$
	\item $0$  $0$  $0$  $0$  $0$  $-1$  $0$  $0$  $0$  $2$  $0$  $-1$  $0$  $0$
	\item $0$  $0$  $0$  $0$  $0$  $-1$  $0$  $1$  $0$  $1$  $-1$  $0$  $0$  $0$
	\item $0$  $0$  $0$  $0$  $0$  $0$  $-1$  $0$  $0$  $2$  $0$  $-1$  $0$  $0$
	\item $0$  $0$  $0$  $0$  $0$  $0$  $-1$  $0$  $1$  $1$  $-1$  $0$  $0$  $0$
	\item $0$  $0$  $0$  $0$  $0$  $0$  $0$  $-1$  $0$  $0$  $2$  $0$  $-1$  $0$
	\item $0$  $0$  $0$  $0$  $0$  $0$  $0$  $0$  $-1$  $0$  $2$  $0$  $-1$  $0$
	\item $0$  $0$  $0$  $0$  $0$  $0$  $0$  $0$  $0$  $-1$  $0$  $2$  $-1$  $0$
	\item $0$  $0$  $0$  $0$  $0$  $0$  $0$  $0$  $0$  $-1$  $1$  $1$  $-1$  $0$
	\item $0$  $0$  $0$  $0$  $0$  $0$  $0$  $0$  $0$  $0$  $-1$  $0$  $2$  $-1$
	\item $0$  $0$  $0$  $0$  $0$  $0$  $0$  $0$  $0$  $0$  $0$  $-1$  $2$  $-1$
	\item $0$  $0$  $0$  $0$  $0$  $0$  $0$  $0$  $0$  $0$  $0$  $0$  $-1$  $1$
	\item $0$  $2$  $0$  $-1$  $0$  $0$  $0$  $0$  $0$  $0$  $0$  $0$  $0$  $0$
	\item $1$  $1$  $-1$  $0$  $0$  $0$  $0$  $0$  $0$  $0$  $0$  $0$  $0$  $0$
\end{itemize}

V-representation:
\begin{itemize}
	\item $U_{1,7}$, $U_{2,7}$, $U_{3,7}$, $U_{4,7}$,$U_{5,7}$, $U_{6,7}$;
	\item $1$  $0$  $1$  $0$  $1$  $0$  $0$  $1$  $1$  $0$  $1$  $0$  $1$  $1$
	\item $0$  $1$  $1$  $2$  $2$  $3$  $3$  $3$  $3$  $4$  $4$  $5$  $5$  $6$
	\item $6$  $1$  $6$  $2$  $6$  $3$  $3$  $6$  $6$  $4$  $6$  $5$  $6$  $6$
	\item $2$  $1$  $2$  $2$  $2$  $2$  $2$  $2$  $2$  $2$  $2$  $2$  $2$  $2$
	\item $5$  $1$  $5$  $2$  $5$  $3$  $3$  $5$  $5$  $4$  $5$  $5$  $5$  $5$
	\item $3$  $1$  $3$  $2$  $3$  $3$  $3$  $3$  $3$  $3$  $3$  $3$  $3$  $3$
	\item $4$  $1$  $5$  $2$  $6$  $3$  $3$  $6$  $6$  $4$  $6$  $5$  $6$  $6$
	\item $0$  $1$  $1$  $2$  $2$  $2$  $2$  $2$  $2$  $2$  $2$  $2$  $2$  $2$
	\item $3$  $1$  $4$  $2$  $5$  $3$  $3$  $5$  $5$  $4$  $5$  $5$  $5$  $5$
	\item $5$  $1$  $6$  $2$  $6$  $3$  $3$  $6$  $6$  $4$  $6$  $5$  $6$  $6$
	\item $4$  $1$  $5$  $2$  $5$  $3$  $3$  $5$  $5$  $4$  $5$  $5$  $5$  $5$
	\item $0$  $1$  $1$  $1$  $1$  $1$  $1$  $1$  $1$  $1$  $1$  $1$  $1$  $1$
	\item $2$  $1$  $3$  $2$  $3$  $3$  $3$  $3$  $3$  $3$  $3$  $3$  $3$  $3$
	\item $3$  $1$  $4$  $2$  $5$  $3$  $3$  $6$  $6$  $4$  $6$  $5$  $6$  $6$
	\item $2$  $1$  $3$  $2$  $4$  $3$  $3$  $5$  $5$  $4$  $5$  $5$  $5$  $5$
	\item $0$  $1$  $1$  $2$  $2$  $3$  $3$  $3$  $3$  $3$  $3$  $3$  $3$  $3$
	\item $6$  $2$  $6$  $4$  $6$  $5$  $6$  $6$  $6$  $6$  $6$  $6$  $6$  $6$
	\item $4$  $1$  $4$  $2$  $4$  $3$  $3$  $4$  $4$  $4$  $4$  $4$  $4$  $4$
	\item $2$  $2$  $4$  $4$  $6$  $5$  $6$  $6$  $6$  $6$  $6$  $6$  $6$  $6$
	\item $2$  $1$  $3$  $2$  $4$  $3$  $3$  $4$  $4$  $4$  $4$  $4$  $4$  $4$
	\item $4$  $2$  $6$  $4$  $6$  $5$  $6$  $6$  $6$  $6$  $6$  $6$  $6$  $6$
	\item $3$  $1$  $4$  $2$  $4$  $3$  $3$  $4$  $4$  $4$  $4$  $4$  $4$  $4$
	\item $3$  $2$  $4$  $4$  $5$  $5$  $6$  $6$  $6$  $6$  $6$  $6$  $6$  $6$
	\item $1$  $2$  $3$  $4$  $5$  $5$  $6$  $6$  $6$  $6$  $6$  $6$  $6$  $6$
	\item $2$  $2$  $4$  $4$  $5$  $5$  $6$  $6$  $6$  $6$  $6$  $6$  $6$  $6$
	\item $6$  $3$  $8$  $6$  $10$  $8$  $9$  $12$  $11$  $10$  $12$  $11$  $12$  $12$
	\item $4$  $3$  $7$  $6$  $10$  $8$  $9$  $12$  $11$  $10$  $12$  $11$  $12$  $12$
	\item $5$  $3$  $8$  $6$  $10$  $8$  $9$  $12$  $11$  $10$  $12$  $11$  $12$  $12$
	\item $4$  $3$  $6$  $6$  $8$  $8$  $9$  $10$  $9$  $10$  $10$  $10$  $10$  $10$
	\item $2$  $2$  $4$  $4$  $6$  $6$  $6$  $8$  $7$  $8$  $8$  $8$  $8$  $8$
	\item $2$  $3$  $5$  $6$  $8$  $8$  $9$  $10$  $9$  $10$  $10$  $10$  $10$  $10$
	\item $3$  $3$  $6$  $6$  $8$  $8$  $9$  $10$  $9$  $10$  $10$  $10$  $10$  $10$
	\item $5$  $3$  $7$  $6$  $9$  $8$  $9$  $11$  $10$  $10$  $11$  $11$  $11$  $11$
	\item $3$  $3$  $6$  $6$  $9$  $8$  $9$  $11$  $10$  $10$  $11$  $11$  $11$  $11$
	\item $4$  $3$  $7$  $6$  $9$  $8$  $9$  $11$  $10$  $10$  $11$  $11$  $11$  $11$
	\item $6$  $2$  $8$  $4$  $10$  $6$  $6$  $11$  $12$  $8$  $12$  $10$  $12$  $12$
	\item $3$  $2$  $5$  $4$  $7$  $6$  $6$  $8$  $9$  $7$  $9$  $8$  $9$  $9$
	\item $4$  $2$  $6$  $4$  $8$  $6$  $6$  $9$  $10$  $8$  $10$  $10$  $10$  $10$
	\item $2$  $2$  $4$  $4$  $6$  $6$  $6$  $7$  $8$  $7$  $8$  $8$  $8$  $8$
	\item $0$  $2$  $2$  $4$  $4$  $5$  $6$  $5$  $6$  $6$  $6$  $6$  $6$  $6$
	\item $2$  $2$  $4$  $4$  $6$  $6$  $6$  $7$  $8$  $8$  $8$  $8$  $8$  $8$
	\item $0$  $1$  $1$  $2$  $2$  $3$  $3$  $3$  $3$  $4$  $4$  $4$  $4$  $4$
	\item $6$  $3$  $8$  $6$  $10$  $9$  $8$  $11$  $12$  $10$  $12$  $11$  $12$  $12$
	\item $4$  $3$  $7$  $6$  $10$  $9$  $8$  $11$  $12$  $10$  $12$  $11$  $12$  $12$
	\item $5$  $3$  $8$  $6$  $10$  $9$  $8$  $11$  $12$  $10$  $12$  $11$  $12$  $12$
	\item $4$  $3$  $6$  $6$  $8$  $9$  $8$  $9$  $10$  $10$  $10$  $10$  $10$  $10$
	\item $1$  $2$  $3$  $4$  $5$  $6$  $6$  $6$  $7$  $7$  $7$  $7$  $7$  $7$
	\item $2$  $3$  $5$  $6$  $8$  $9$  $8$  $9$  $10$  $10$  $10$  $10$  $10$  $10$
	\item $3$  $3$  $6$  $6$  $8$  $9$  $8$  $9$  $10$  $10$  $10$  $10$  $10$  $10$
	\item $5$  $3$  $7$  $6$  $9$  $9$  $8$  $10$  $11$  $10$  $11$  $11$  $11$  $11$
	\item $3$  $3$  $6$  $6$  $9$  $9$  $8$  $10$  $11$  $10$  $11$  $11$  $11$  $11$
	\item $4$  $3$  $7$  $6$  $9$  $9$  $8$  $10$  $11$  $10$  $11$  $11$  $11$  $11$
	\item $2$  $1$  $3$  $2$  $4$  $3$  $3$  $5$  $5$  $4$  $6$  $5$  $6$  $6$
	\item $6$  $2$  $6$  $4$  $6$  $6$  $5$  $6$  $6$  $6$  $6$  $6$  $6$  $6$
	\item $2$  $2$  $4$  $4$  $6$  $6$  $5$  $6$  $6$  $6$  $6$  $6$  $6$  $6$
	\item $4$  $2$  $6$  $4$  $6$  $6$  $5$  $6$  $6$  $6$  $6$  $6$  $6$  $6$
	\item $3$  $2$  $4$  $4$  $5$  $6$  $5$  $6$  $6$  $6$  $6$  $6$  $6$  $6$
	\item $1$  $2$  $3$  $4$  $5$  $6$  $5$  $6$  $6$  $6$  $6$  $6$  $6$  $6$
	\item $2$  $2$  $4$  $4$  $5$  $6$  $5$  $6$  $6$  $6$  $6$  $6$  $6$  $6$
	\item $3$  $2$  $5$  $4$  $7$  $6$  $6$  $9$  $8$  $7$  $9$  $8$  $9$  $9$
	\item $6$  $2$  $8$  $4$  $10$  $6$  $6$  $12$  $11$  $8$  $12$  $10$  $12$  $12$
	\item $2$  $2$  $4$  $4$  $6$  $6$  $6$  $8$  $7$  $7$  $8$  $8$  $8$  $8$
	\item $4$  $2$  $6$  $4$  $8$  $6$  $6$  $10$  $9$  $8$  $10$  $10$  $10$  $10$
	\item $0$  $2$  $2$  $4$  $4$  $6$  $5$  $6$  $5$  $6$  $6$  $6$  $6$  $6$
	\item $1$  $2$  $3$  $4$  $5$  $6$  $6$  $7$  $6$  $7$  $7$  $7$  $7$  $7$
	\item $0$  $1$  $1$  $2$  $2$  $3$  $3$  $3$  $3$  $4$  $4$  $5$  $5$  $5$
\end{itemize}
\subsection{The Hasse diagram of $P_6$ and $P_7$}
\begin{figure*}[htp]
	\centering
	\vspace{-1cm}
	\rotatebox{90}{\includegraphics[width=1.3\linewidth,height=0.95\linewidth]{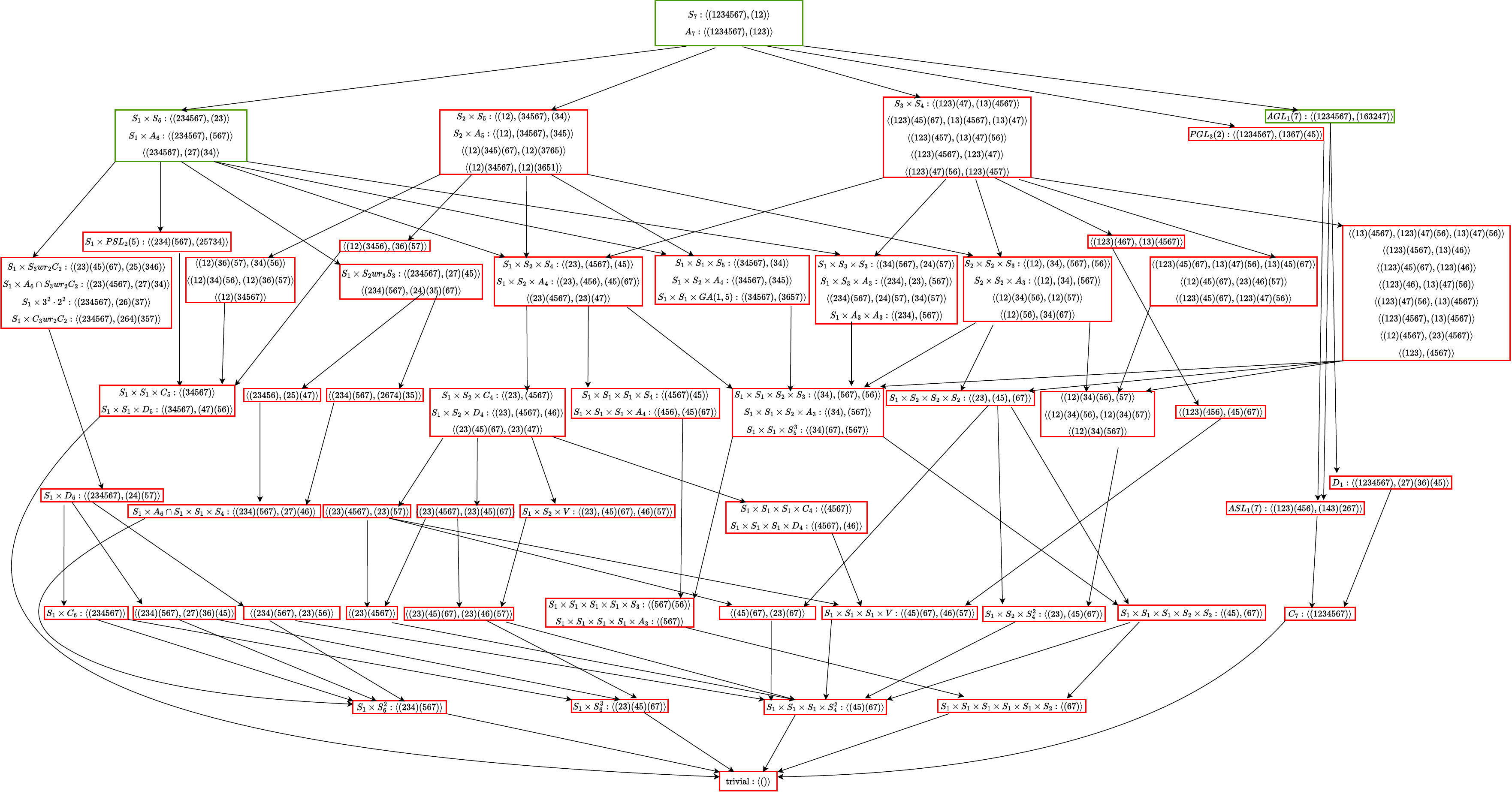}}
	\captionsetup{font=small}
	\caption{The Hasse diagram of $P_7$}
	\label{fig5}
\end{figure*}
\IEEEtriggeratref{4}

\end{document}